\documentclass[12pt]{article}

\usepackage[top=1.25in, bottom=1.25in, left=1.25in, right=1.25in]{geometry}

\usepackage{amssymb,amsfonts,amsthm,amsmath, color}
\usepackage{mathrsfs,pifont}
\usepackage{graphicx}

\usepackage[font={small,sl}]{caption}
\usepackage[all]{xy}

\usepackage{hyperref}
\usepackage[backrefs,msc-links]{amsrefs}

\newcommand{\C}{\mathbb{C}}
\newcommand{\Z}{\mathbb{Z}}
\newcommand{\R}{\mathbb{R}}
\newcommand{\N}{\mathbb{N}}
\newcommand{\Q}{\mathbb{Q}}

\newcommand{\bG}{\mathsf{G}}
\newcommand{\bZ}{\mathsf{Z}}
\newcommand{\bv}{\mathsf{v}}

\newcommand{\vt}{\,\widetilde{\mathsf{v}}\,}
\newcommand{\bP}{\mathbb{P}}
\newcommand{\bT}{\mathsf{T}}
\newcommand{\bJ}{\mathsf{J}}

\newcommand{\bfG}{\mathbf{G}}

\newcommand{\bA}{\mathsf{A}}
\newcommand{\bU}{\mathsf{U}}
\newcommand{\cK}{\mathscr{K}}
\newcommand{\cL}{\mathscr{L}}
\newcommand{\cO}{\mathscr{O}}
\newcommand{\cT}{\mathscr{T}}

\newcommand{\cB}{\mathscr{B}}
\newcommand{\cV}{\mathscr{V}}
\newcommand{\cM}{\mathcal{M}}
\newcommand{\cU}{\mathscr{U}}

\newcommand{\cZ}{\mathscr{Z}}
\newcommand{\cW}{\mathscr{W}}

\newcommand{\fp}{\mathfrak{p}}
\newcommand{\fn}{\mathfrak{n}}

\newcommand{\fX}{\mathfrak{X}}
\newcommand{\fR}{\mathfrak{R}}
\newcommand{\bw}{\mathsf{w}}
\newcommand{\vac}{\mathsf{vac}}

\newcommand{\aroof}{\widehat{\mathsf{a}}}
\newcommand{\Ah}{\widehat{A}}
\newcommand{\cUl}{\cU_\hbar\big(\widehat{\widehat{\mathfrak{gl}}}_\ell\big)}

\newcommand{\bPi}{\boldsymbol{\Pi}}

\newcommand{\bdiota}{\boldsymbol{\iota}} 

\newcommand{\Hd}{{H}^{\raisebox{0.5mm}{$\scriptscriptstyle \bullet$}}}

\DeclareMathOperator{\rk}{rk}
\DeclareMathOperator{\cent}{center}

\DeclareMathOperator{\Obs}{Obs}
\DeclareMathOperator{\Lie}{Lie}

\DeclareMathOperator{\Aut}{Aut}
\DeclareMathOperator{\Hom}{Hom}
\DeclareMathOperator{\End}{End}
\DeclareMathOperator{\Pic}{Pic}
\DeclareMathOperator{\pt}{pt}
\DeclareMathOperator{\tr}{tr}
\DeclareMathOperator{\ev}{ev}
\DeclareMathOperator{\tev}{\widetilde{ev}}
\DeclareMathOperator{\Stab}{Stab}
\DeclareMathOperator{\Unstab}{Unstab}
\DeclareMathOperator{\Attr}{Attr}
\DeclareMathOperator{\Spec}{Spec}
\DeclareMathOperator{\Rep}{Rep}

\DeclareMathOperator{\Ker}{Ker}
\DeclareMathOperator{\supp}{supp}
\DeclareMathOperator{\Hilb}{Hilb}

\newcommand{\Ld}{{\Lambda}^{\!\raisebox{0.5mm}{$\scriptscriptstyle
      \bullet$}}\!}
\newcommand{\Ldi}{{\Lambda}^{\!\raisebox{0.5mm}{$
\scriptstyle \diamond$}}}
\newcommand{\be}{\mathbf{e}}

\newcommand{\QM}{\mathsf{QM}}
\newcommand{\Ct}{\mathbb{C}^\times}

\newcommand{\tO}{\widehat{\mathscr{O}}}
\newcommand{\vir}{\textup{vir}}

\newcommand{\cF}{\mathscr{F}}

\newcommand{\fh}{\mathfrak{h}} 
\newcommand{\fg}{\mathfrak{g}} 
\newcommand{\fgh}{\widehat{\mathfrak{g}}} 
\newcommand{\fhh}{\widehat{\mathfrak{h}}} 
\newcommand{\fgl}{\widehat{\mathfrak{gl}}}

\newcommand{\fb}{\mathbf{f}}
\newcommand{\bp}{\mathbf{p}}
\newcommand{\bs}{\mathbf{s}}
\newcommand{\Be}{\mathfrak{Be}}
\newcommand{\Zb}{\overline{\bZ}}

\newtheorem{Proposition}{Proposition} 
\newtheorem{Lemma}{Lemma} 
 
\newtheorem{Theorem}{Theorem}

\newtheorem{Definition}{Definition} 

\newcommand{\rdd}{/\!\!/}

\begin{document}

\title{Quasimap counts and Bethe eigenfunctions}
\author{Mina Aganagic and Andrei Okounkov}
\date{} 
\maketitle

\abstract{
We associate an explicit equivalent 
descendent insertion to any 
relative insertion in quantum K-theory of Nakajima varieties. 

This also serves as an explicit formula for 
off-shell Bethe eigenfunctions for general quantum loop 
algebras associated to quivers and gives the 
general integral solution to the corresponding quantum 
Knizhnik-Zamolodchikov and dynamical 
$q$-difference equations.}

\section{Introduction}

\subsection{Overview} 

\subsubsection{} 

The problem solved in this paper has a 
representation-theoretic side and a geometric side. 

In representation theory of quantum affine algebras, and its 
applications to exactly solvable models of mathematical physics, 
a very important role is played by certain $q$-difference 
equations. These are the quantum Knizhnik-Zamolodchikov 
equations (qKZ),
see \cites{FR,EFK} and the corresponding commuting dynamical equations
\cites{EV1,EV,FMTV,FTV,TV,TV3}. A lot of 
research has been focused on solving these equations by 
integrals of Mellin-Barnes type, see e.g.
\cites{EFK,Matsuo, Resh,TV1,TV2,TV4,TV5}. 
Such integrals, in particular,
give explicit formulas for
Bethe eigenvectors in the stationary phase $q\to 1$ 
limit. Here we give a general integral solution for 
tensor products of evaluation representations of quantum 
affine Lie algebras associated to quivers as in \cite{MO}. 
These include, in particular, double loop algebras of the form 
$\cUl$, which are known under many different names and play 
a very important role in many branches of modern mathematical 
physics, see \cite{NeTh} for a detailed introduction and further 
references. 

For us, these representation-theoretic problems are reflections of
certain geometric questions about enumerative 
K-theory of quasimaps to Nakajima quiver varieties
(see \cite{pcmi} for an introduction). In mathematical 
physics, Nakajima varieties appear in 
supersymmetric gauge theories as Higgs branches of moduli of 
vacua, and K-theoretic quasimaps counts may be interpreted
as Higgs branch computations of 3-dimensional 
supersymmetric indices\footnote{while the Mellin-Barnes 
integrals may be interpreted as the equivalent 
Coulomb branch computations, see e.g.\ \cite{afo} for 
further discussion.}. Nekrasov and Shatashvili 
\cites{NS1,NS2} were the first to make the connection between 
these indices and Bethe equations, see also 
\cite{NekVid1}. The actual problem solved
here is to associate an explicit equivalent descendent insertion to any 
relative insertion in enumerative K-theory of 
quasimaps to Nakajima varieties, see below and \cite{pcmi} for 
an explanation of these terms. 

Our results are
complementary to the recent important work of Smirnov 
\cite{S2} who associates an equivalent relative insertion 
to any descendent insertion in terms of a certain graphical 
calculus and canonical tensors associated to the quantum group.
Here we allow a wider supply of descendent insertions, and 
get a simple formula (with an arguably simpler proof) for a map 
going in the opposite direction.

For quivers of affine ADE type, 
quasimap counts compute the K-theoretic Donaldson-Thomas
invariants of threefolds fibered in ADE 
surfaces\footnote{Those include local curves, that is, threefolds
fibered in $A_0 = \C^2$.}. Finding an equivalence between 
relative and descendent insertions in Donaldson-Thomas 
theories of threefolds is a well-known problem of crucial technical 
importance for the developments of the theory, see 
\cite{mnop2} for an early discussion and \cites{PP1,PP2} for 
major further progress in cohomology. 
Our formulas are both more explicit and work in 
K-theory\footnote{Equivariant K-theory is similarly the natural setting of
  Smirnov's formulas \cite{S2}.}.

\subsubsection{}

Let $\fg$ be a Lie algebra associated to a quiver 
with a vertex set $I$ 
as in \cite{MO}. For example, modulo center, $\fg$ is the corresponding 
simple Lie algebra for quivers of finite ADE type and 
$\fg=\widehat{\mathfrak{gl}}_\ell$ for the cyclic quiver 
$\widehat{A}_{\ell-1}$ with $\ell$ vertices. 

Extending the work of 
Nakajima \cite{Nak3}, tensor products of fundamental evaluation 
representations $F_i(a)$, $i\in I$, of the the corresponding 
quantum loop algebra $\cU_\hbar(\fgh)$ may be realized 
geometrically using equivariant K-groups of Nakajima quiver 
varieties \cites{MO,OS}.

Let $X=\cM(\bv,\bw)$ be a Nakajima variety indexed
by dimension vectors $\bv,\bw \in \N^I$ and let $\bT$ be a
 torus of automorphisms of $X$. It scales the canonical 
symplectic form $\omega$ on $X$ and
$$
\hbar = \textup{weight of $\omega$} \in K_\bT(\pt) 
$$ 
is the deformation parameter in $\cU_\hbar(\fgh)$. We set 
$\bA = \Ker \hbar$ and assume that $\bA$ contains the 
torus 
\begin{equation}
\bA \supset \left\{\begin{pmatrix}
  a_{i1} \\
& \ddots \\
&& a_{i\bw_i} 
\end{pmatrix} \right\} 
\subset \prod GL(W_i) \subset \Aut(X) \label{torus A}
\end{equation}
acting on the framing spaces $W_i$ of the 
quiver.

A certain integral 
form of $\cU_\hbar(\fgh)$ acts by correspondences between 
equivariant K-theories of Nakajima varieties so that 
\begin{equation}
  \label{KXw}
  K_\bT(X) \otimes_{K_\bT(\pt)} \textup{field} \cong 
\left( \bigotimes_{i\in I} \bigotimes_{j=1}^{\bw_i}
  F_i(a_{i,j})\right)
_{\textup{weight}=\bv}
\end{equation}
where the weight is with respect to the Cartan subalgebra 
$\fh \subset \fg$ acting by linear
 function of $\bw$ and $\bv$.

\subsubsection{}

Quantum Knizhnik-Zamolodchikov equations of I.~Frenkel and
N.~Reshetikhin \cites{FR,EFK} are certain canonical $q$-difference equations 
for a function of the variables $a_{ij}$ in \eqref{torus A} 
with values in the vector space \eqref{KXw}. 
The shift $q\in \Ct$ here is a free parameter related to the loop-rotation 
automorphism of $\cU_\hbar(\fgh)$. In the original setup of 
\cite{FR}, qKZ equations appeared as difference equation for 
conformal blocks of $\cU_\hbar(\fgh)$ at a fixed level and there was 
a relation between $q$, the deformation parameter
$\hbar$, and the level. 
The geometric meaning of $q$ will be explained below. 

As a
parameter, qKZ equations take 
$$
z \in \bZ=\textup{group-like elements of $\cU_\hbar(\fgh)$} \, \Big/ \, 
\textup{center} 
$$
or, equivalently, of the torus corresponding to the $\bv$-part in
$\fh$. The monomials $z^\bv$ are the characters $\bZ$.

Compatible systems of $q$-difference equations in $z$ were 
studied in detail by Etingof, Tarasov, Varchenko, and others in 
the case of finite-dimensional
 algebras $\fg$, see \cite{TV} and 
also for example \cites{EV1,EV,FMTV,FTV,TV1,TV2,TV3,TV4,TV5}. 
In particular, for finite-dimensional $\fg$, the 
commuting equations were understood in terms of the lattice part in 
the dynamical quantum affine Weyl group of $\cU_\hbar(\fgh)$ in
\cite{EV}. These \emph{dynamical} difference equations 
 are intrinsic to $\cU_\hbar(\fgh)$ and make sense in 
an arbitrary weight space even in the absence of tensor product 
structure and associated qKZ equations. 

For general $\fg$, the dynamical difference equations were constructed
in \cite{OS}. 

\subsubsection{}

First Chern classes of tautological bundles give a natural map 
$$
\Z^I \to H^2(X,\Z)
$$
which is known to be surjective \cite{mn}. The dual map sends the group algebra of $H_2(X,\Z)$ to $\C[\bZ]$ 
and makes the monomials $z^\bv$ degree
labels for curve counts in $X$. The 
variables $z$ are known as the K\"ahler variables
for $X$ in the parlance of enumerative
geometry. The so-called K\"ahler moduli space is, in the case of 
Nakajima varieties, a certain toric compactification $\Zb \supset
\bZ$. 

 With the identification
\eqref{KXw}, the qKZ and dynamical equations become the 
quantum difference equations in enumerative K-theory of quasimaps 
to $X$ \cites{pcmi,OS}. 
These $q$-difference equations shift the 
equivariant variables $a$ and the K\"ahler variables $z$ by the 
fundamental weight $q$ of the group 
$$
\Ct_q = \Aut(\bP^1,0,\infty) 
$$
that acts on the moduli spaces of quasimaps 
$$
\QM(X) = 
\left\{ 
f: \bP^1 \dasharrow X\right\} \big/ \cong 
$$ by automorphisms 
of the domain, see \cite{pcmi} for an introduction.  

\subsubsection{}

While the natural evaluation map 
$$
\QM(X) \owns f \mapsto (f(0),f(\infty)) \in \fX\times \fX 
$$
only goes to the stack quotient 
\begin{equation}
\fX = 
\left[
\frac{\textup{prequotient}}{G}
\right] \supset 
\frac{\textup{stable locus}}{G}=  X\label{Xstack},
\end{equation}
one can impose constraints on $f$ or modify the moduli spaces to turn 
enumerative counts into correspondences on $X$, or correspondences
between $\fX$ and $X$. Conditions imposed at $0,\infty
\in \bP^1$ are customary called \emph{insertions}, just like
insertions in functional integrals.

K-theoretic counts of quasimaps with different insertions at
$0,\infty \in \bP^1$ give objects  
of different nature as functions of $a$, $z$, and other parameters. 
For certain insertions, we get  a fundamental 
solutions of the quantum difference equations, while for other 
insertions we get integrals of 
Mellin-Barnes type. 

\subsubsection{}
By an integral of Mellin-Barnes type we mean an integral of the form 
\begin{equation}
I_{\alpha\beta}(z,\dots) = \int_{\gamma \subset T_G/W_G} 
 \fb_\alpha(x) \, \mathbf{g}_\beta(x) \, \be(x,z) \, 
\prod \frac{\phi(x^{\lambda_i} b_i)}{\phi(x^{\lambda_i} c_i)} \,
\prod \frac{dx_k}{2\pi i x_k}  \label{IMB}
\end{equation}
up to multiplicative shift\footnote{The exact 
form of this multiplicative shift, which is of no importance 
here, is discussed in the Appendix.} in $z$, where
\begin{itemize}
\item the integration is over a middle-dimensional cycle in the
  quotient of a torus $T_G$ by a finite group $W_G$. Concretely, 
$T_G\subset G$ is a maximal torus of the group $G$ in \eqref{Xstack}, 
with Weyl group $W_G$. Geometrically, the coordinates on $T_G/W_G$ are
the characteristic classes of the universal bundles on $X$. Since
these are known to span the K-theory of $X$ \cite{mn}, we have a natural 
embedding
\begin{equation}
\label{SpecK}
\xymatrix{
\Spec K_\bT(X) \ar@{^(->}^{\iota}[r]  & \bT \times T_G/W_G  \ar[d]_{\pi_\bT} \\
& \bT 
}
\end{equation}
finite over the torus $\bT$ of equivariant parameters. The variables
in $\bT$ including $\hbar$ and $a$ are parameters in \eqref{IMB}
and the integral should be viewed as an integral in the fibers of 
the projection $\pi_\bT$. 
\item 
the cycle $\gamma$ extracts the residues of the integrand at  
$q$-translates of the pole at the image of $\iota$ in \eqref{SpecK}. 
\item the function 
$$
\phi(y) = \prod_{n=0}^\infty (1-q^n y) 
$$
solves the simplest $q$-difference equation and replaces the
reciprocal of the $\Gamma$-function in the $q$-world. Ratios of 
the form 
$\frac{\phi(x^{\lambda} b)}{\phi(x^{\lambda} c)}$ generalize
complex powers of linear forms ubiquitous in hypergeometric 
integrals. Instead of hyperplanes, we have translates of codimension
1 subtori in $\bT\times T_G$. 
\item
the weights $\lambda_i$ and the shifts $b_i,c_i$
 involve the roots of $G$ and the weights of 
$\bT\times G$ action on the 
prequotient in \eqref{Xstack}. For Nakajima varieties, 
\eqref{Xstack} is an algebraic symplectic reduction of a 
cotangent bundle, and  the self-duality of this setup 
implies 
$$
\{b_i,c_i\} = \{q t^{\nu_i}, \hbar t^{\nu_i}\} 
$$
for a certain weight $t^{\nu_i}$ of $\bT$ on the prequotient in 
\eqref{Xstack}. 
\item 
the function 
\begin{equation}
  \label{be}
  \be(x,z) = \exp \left( (\ln q)^{-1} \sum_{i,k} \ln x_{i,k} \, \ln
    z_i \right) 
\end{equation} 
where the coordinate $x_{i,k}$ are grouped according to 
$G = \prod_{i\in I} GL(\bv_i)$ solves monomial 
$q$-difference equations in $x$ and $z$ 
and makes the integral \eqref{IMB} 
a $q$-difference analog of Fourier or Mellin transform. 
\item
the function $\mathbf{g}_\beta(x)$ is an elliptic function on $x$ 
(that is, a constant, from the viewpoint of $q$-difference equations) 
regular at the location of $\gamma$. It is convenient to use a 
suitable basis of such functions as a mechanism to generate a basis 
in the $\rk K(X)$-dimensional space of solutions of the quantum 
difference equations. 

{}From the perspective of \cite{afo,ese}, see in particular 
Section 6.2 in \cite{ese} and 
Section 5.4 in \cite{afo} for detailed examples, it is natural to use \emph{elliptic} stable envelopes to build functions 
$\mathbf{g}_\beta(x)$. Our focus in this paper, however, is on the
functions $\fb_\alpha(x)$, and their relations to \emph{K-theoretic} 
stable envelopes.  
\item 
the \emph{Bethe subscheme} 
$$
\Be = \left\{\frac{\partial}{\partial x} \cW =0 \right\}  
\subset \Zb \times \bT \times T_G/W_G 
$$
where\footnote{The function $\cW$ is known as the Yang-Yang function.}  
\begin{equation}
\cW = \lim_{q\to 1} \ln(q) \ln \left( \be(x,z) \, 
\prod \frac{\phi(x^{\lambda_i} b_i)}{\phi(x^{\lambda_i}
  c_i)}\right)\label{cW}
\end{equation}
appears as the critical points of the integral in the $q\to 1$ limit. 
It is the joint spectrum of the corresponding commuting operators on 
$K_\bT(X)$ and the map 
$$
K_\bT(X) \owns \alpha \mapsto \fb_\alpha \mapsto \C[\Be]
$$
gives the Jordan normal form of the $\C[\Be]$-action on $K_\bT(X)$. 
The fiber of $\Be$ over $0\in \Zb$ is the spectrum of 
K-theory of $X$ in \eqref{SpecK}. The concrete form of Bethe equations is recalled in the 
Appendix.

The connection between Bethe equations and quiver
gauge theories whose Higgs branch is $X$ is one of the main points of a very influential 
sequence of papers by Nekrasov and Shatashvili, see \cites{NS1,NS2}. 
\item
finally, the function $\fb_\alpha(x)$ is a rational function of $x$
that depends linearly on $\alpha\in K_\bT(X)$ and restricts to 
$\alpha$ on the image of $\iota$ in \eqref{SpecK}. It is known 
under various names including ``off-shell Bethe eigenfunction'' 
and ``weight function''. This function $\fb_\alpha(x)$ will be 
the most important player in this paper. 
\end{itemize}

Partition functions of supersymmetric gauge theories can be often expressed as
integrals of the general form \eqref{IMB}, see e.g.\ \cite{MNS,Ninst} 
for prominent examples of such computation. The group $G$ in this
case is the complexification of the gauge group
and integration corresponds, via Weyl integration formula, to
extracting invariants of constant gauge 
transformations.\footnote{Alternatively, the quotient $T_G/W_G$ is 
closely related to the Coulomb branch of vacua of the theory and 
the integral \eqref{IMB} may be interpreted as an equivalent direct 
computation on the Coulomb branch.} See e.g.\ 
\cite{afo} for and introductory mathematical discussion and an
explanation of how  integrals 
of the form \eqref{IMB} appear in enumerative 
theory of quasimaps to $X$ with \emph{descendent} insertions. 
See also e.g.\ 
\cite{PSZ} for a detailed discussion of the Nekrasov-Shatashvili 
connection between Bethe 
equation and enumerative theory of quasimaps that does not make 
an explicit use of Mellin-Barnes integrals.

\subsubsection{}

The space of possible descendent insertions at $0\in \bP^1$ 
$$
\left\{\begin{matrix}
   \textup{descendent}\\
\textup{insertions} 
  \end{matrix} \right\} = K_{\bT \times G}(\pt) = \Z[\bT \times T_G/W_G] 
$$
corresponds to all possible Laurent polynomials $f_\alpha(x)$ 
in \eqref{IMB}. A choice of $\mathbf{g}_\beta$
 corresponds to a \emph{nonsingular} insertion at $\infty\in \bP^1$. There is a third flavor of insertions, called \emph{relative} and they 
take a class $\alpha\in K_\bT(X)$ as an input. This is explained in
Section \ref{s_insertions} and, in more details, in \cite{pcmi}.

By a geometric 
argument, K-theoretic count of quasimaps with a relative insertion 
at $0$ and a nonsingular insertion at $\infty$ gives a fundamental 
solution of the quantum difference equations, see Section 8 in
\cite{pcmi} for details.

\subsubsection{}
In this paper, we will describe a linear map 
\begin{equation}
  \label{afa}
  \left\{\begin{matrix}
    \textup{relative}\\
\textup{insertions} 
  \end{matrix} \right\}= 
K_\bT(X) \owns \alpha \mapsto \fb_\alpha \in \Q(\bT \times T_G/W_G) 
= \left\{\begin{matrix}
\textup{localized}\\
    \textup{descendent}\\
\textup{insertions} 
  \end{matrix} \right\}
\end{equation}
that \emph{preserves} K-theoretic counts, and therefore makes the 
Mellin-Barnes integral \eqref{IMB} a solution of the quantum 
difference equations. Among all quasimaps, there are 
degree zero, that is, constant quasimaps, which means 
\begin{equation}
  \label{fbres}
  \fb_\alpha \big|_{K_\bT(X)} = \alpha 
\end{equation}
in the diagram \eqref{SpecK}. 

In \eqref{afa}, we allow only very specific denominators 
\begin{equation}
\fb_\alpha = \frac{\bs_\alpha}{\Delta_\hbar} \,, 
\quad \bs_\alpha \in \Z[\bT \times T_G/W_G]\,, \label{fsa}
\end{equation}
where $\Delta_\hbar$ is the Koszul complex for the 
moment map equations for $X$, that is, 
\begin{align}
 \Delta_\hbar &= \sum_{k} (-\hbar)^k \Lambda^k \Lie(G)  
\notag \\
   &= \prod_i \prod_{k,l} (1- \hbar x_{i,k}/x_{i,l})
 \label{Dh}
 \end{align}
with the coordinates $x_{i,k}$ grouped as in \eqref{be}.

The numerator $\bs_\alpha$ of $\fb_\alpha$
 is such 
that the counts are still defined in \emph{integral}, that is, 
nonlocalized K-theory. This integrality is crucial and the 
geometric 
mechanism responsible for it will be explained in 
Section \ref{s_int}. In particular, we will make precise the 
mechanism of restriction \eqref{fbres} of a rational 
function to a locus that may be contained in the divisor 
of poles.

\subsubsection{}

The denominator in the correspondence \eqref{afa} is what 
differentiates our approach from other results in the literature, 
notably from a very general result of Smirnov \cite{S2} who gives 
a map 
\begin{equation}
\left\{\begin{matrix}
    \textup{descendent}\\
\textup{insertions} 
  \end{matrix} \right\} 
\to 
K_\bT(X)\otimes \Q(z,q)  = 
\left\{\begin{matrix}
    \textup{relative}\\
\textup{insertions} 
  \end{matrix} \right\} \otimes \Q(z,q)\label{Smi}
\end{equation}
which preserves K-theoretic counts. Restricted to $z=0$, 
the map \eqref{Smi} is the pullback $\iota^*$ in \eqref{SpecK} 
and hence any set of tautological classes that forms a basis of
$K_\bT(X)$ can be used to write integral 
formulas for solutions of quantum difference equations.

\subsubsection{}

Our main result, Theorem \ref{t1} in Section \ref{s_t1},  is an
equivalence 
between a relative
insertion $\alpha$ and the corresponding insertion $\fb_\alpha$ in 
enumerative theory of quasimaps to $X$.  

For $\fb_\alpha$ we give 
a simple formula in terms of K-theoretic 
stable envelopes, see Definition 
\ref{d1} in Section \ref{s_d1}. A representation-theoretic 
translation of this formula is given is \eqref{s=stab} and 
\eqref{fbpart} in Section \ref{sR1}, see also Section \ref{sR2}. 
An introduction 
to K-theoretic stable envelopes may be found in 
\cite{pcmi}. 

An interesting feature of our formula 
for $\fb_\alpha$ is that is does \emph{not} depend on 
variables $z$ or $q$, in marked contrast to \eqref{Smi}.

As an special case, we give explicit formulas for $\fb_\alpha$ for cyclic 
quivers $\Ah_\ell$, that is, for the quantum double loop 
algebras $\cUl$, see Section \ref{cUl}. These formulas can 
be seen as an instance of an abelianization formula for stable 
envelopes in the 
style of \cite{Sh,S1}.  We make the formulas particularly 
explicit in the important case of the Hilbert scheme of points 
in $\C^2$ in Section \ref{s_Hilb}.


\subsubsection{} 

For $\fg=\mathfrak{gl}_n$, our formulas specialize, with a very 
different proof, to integrals studied by Tarasov and Varchenko
\cites{TV,TV1,TV2,TV3,TV4,TV5}.  A connection between what they call the weight function 
and stable envelopes was observed, in this instance, in the papers
\cites{RTV1,RTV2}. These  papers 
were an important source of inspiration for the 
work presented here. 

For $\fg=\widehat{\mathfrak{gl}}_1$, Bethe eigenvectors are obtained 
in \cite{FJMM} in the shuffle algebra realization. Presumably, these
formulas may be extended to $\fg=\widehat{\mathfrak{gl}}_\ell$ using 
e.g.\ the shuffle algebra techniques developed in
\cite{NeTh}. 

Here we don't use any specific features of 
$\widehat{\mathfrak{gl}}_\ell$ and solve a more general problem,
namely the $q$-difference equations that generalize the eigenvalue
problem solved in \cite{FJMM}.

\subsection{Insertions in quantum K-theory}\label{s_insertions}

\subsubsection{}

In enumerative geometry of regular maps $f:C\to X$, it is natural and
important to be able to constrain the values $f(c)$ of $f$ at
specific points 
$c\in C$. For example, the quantum product in $\Hd(X)$ is defined
using counts of 3-pointed rational curves 
$$
f: (C,c_1,c_2,c_3) \to X
$$
such that the points 
$$
(f(c_1),f(c_2),f(c_3))\in X^3
$$ 
meet 3 given cycles in $X$. 

Unlike regular maps, quasimaps may be
singular at a finite set of points of $C$, whence the difficulties
with using the \emph{rational} map 
\begin{equation}
\ev_c: \QM(X) \owns f \dasharrow f(c) \in X \label{evc}
\end{equation}
in enumerative K-theory of the moduli space $\QM(X)$ of stable
quasimaps to $X$. There are at least 3 ways around this
difficulty, namely: 
\begin{itemize}
\item[---] one can restrict to the open set $\QM(X)_\textup{nonsing
    $c$}$ of quasimaps nonsingular at $c$. While the evaluation map is
  not proper on this subset, the equivariant counts are well defined if
$c\in \{0,\infty\} \subset \bP^1\cong C$ and one works equivariantly with
respect to $\Ct_q= \Aut(\bP^1,0,\infty)$. 
\item[---]
one can use a resolution of the map \eqref{evc} 
\begin{equation}
  \label{eq9}
  \xymatrix{
& \QM_\textup{relative $c$}\ar[dr]^{\tev}\\
\QM_\textup{nonsing $c$} \ar[rr]\ar@{^{(}->}[ru]&& X 
}
\end{equation}
provided by the moduli space of quasimaps \emph{relative} the point
$c\in C$. The domain of a relative quasimap is allowed to sprout off a
chain of rational curves joining the new evaluation point $c$ to its
old location on $C$. 
\item[---] 
tautological bundles $\cV_i$ on $C$ are part of the quasimap data
and one can use Schur functors of their fibers at $c$ to impose
constraints on $f(c)$. These are known as 
\emph{descendent insertions} in the
parlance. 
 Recall that Nakajima varieties are constructed as quotients by
$G= \prod GL(V_i)$ and the natural map (sometimes called the 
K-theoretic analog
of the Kirwan map) 
\begin{equation}
K_{G} (\pt) \to K(X) \label{Kir}
\end{equation}
is known to be surjective \cite{mn}. 
Precisely because of the singularities, the 
bundles $\cV_i$ are \emph{not} pulled back from $X$ by $f$ 
and, therefore, descendent insertions \emph{do not} factor through the
Kirwan map \eqref{Kir}. 
\end{itemize}

\subsubsection{}

Since the options listed above express in different precise languages
 the same intuitive idea
of constraining the value $f(c)$, one expects to have a translation
between e.g.\ relative and descendent insertions at $c$. 

This turns
out to be a highly nontrivial problem with important geometric
applications, for instance, in Donaldson-Thomas theory. Early
discussion of it may be found in \cite{mnop2} and, in cohomology, a very
important progress on this problem was achieved by Pandharipande and Pixton in \cites{PP1,PP2}. Geometric representation theory provides a
different and perhaps more powerful approach to these problems, as
demonstrated by A.~Smirnov in \cite{S2}. 

\subsubsection{}

In a fully
equivariant theory, with the action of $\Ct_q$ included, it is
possible to mix and match the type of insertions at the 
$\Ct_q$-fixed points $\{0,\infty\}$ of the domain $C$. 
It is natural to interpret 2-pointed counts as correspondences acting 
on $K(X)$ or as correspondences between $X$ and 
 the stack $\fX$. 

More precisely, one has to localize $K(X)$ in the
presence of nonsingular insertions and work in formal power series in
the variables
\begin{equation}
z^{\deg f}  \in \textup{semigroup algebra of
$H_2(X,\Z)_\textup{effective}$}\label{zdeg}
\end{equation}
that keep track of the degree of a quasimap. These are usually called
the \emph{K\"ahler variables}, as opposed to the the 
\emph{equivariant variables} which include $q$ and 
the coordinates on a maximal torus 
$$
\bT = \bA \times \Ct_\hbar \subset \Aut(X)\,,
$$
where $\hbar$ is the $\bT$-weight of the symplectic form on
$X$ and $\bA = \Ker \hbar$. 

\subsubsection{}\label{s_list_c}

The geometric, representation-theoretic, and functional nature of
the resulting operators strongly depends on the type of insertions chosen, as
illustrated by the following list. In this list we indicate the type
of insertion at 0 followed by the type of insertion at $\infty$. 
Obviously, the roles of $0$ and $\infty$ may be switched by the
automorphism of $\bP^1$ that permutes them and sends $q$ to $q^{-1}$.

\smallskip 

\noindent
\textbf{relative/relative}, also 
known as the \emph{glue} operator $\bfG$, is a
generalization of the longest element in the quantum
dynamical Weyl group of the nonaffine subalgebra 
$$
\cU_\hbar(\mathfrak{g}) \subset 
\cU_\hbar(\fgh)\,,
$$
see \cite{OS} and also \cite{slc}. It does not depend on
$q$ and is a rational function of the K\"ahler variables $z$. It also does not depend the variables $a$ in $\bA$ in certain special bases
of $K_\bT(X)$, see Section 10.3 in \cite{pcmi}. 

\smallskip 

\noindent
\textbf{relative/nonsingular}, also 
known as the \emph{capping} operator $\bJ$, gives a fundamental solution to
$q$-difference equations in both K\"ahler and equivariant 
variables. Difference equations
with respect to $z$ may be interpreted as the action of the lattice
inside the quantum dynamical affine Weyl group of $\cU_\hbar(\fgh)$. 
Difference equations with respect to $a\in\bA$ are the quantum
Knizhnik-Zamolodchikov equations. 

\smallskip 

\noindent
\textbf{descendent/nonsingular} is
also known as the \emph{vertex with descendents}\footnote{The vertex without
descendents refers to having 
no insertions at $0$.}, or the so-called big-I function in the more
conventional nomenclature that goes back to Givental. Its computation by 
$\Ct_q$-localization may be converted into a Mellin-Barnes type
integral over a certain middle-dimensional 
cycle $\gamma$ in a maximal torus of $G$. Such integrals are a
standard practice in SUSY gauge theory literature, and can be also
explained mathematically,  see e.g.\ the
Appendix in \cite{afo}. Descendent
insertions become functions $\fb_\alpha$ in \eqref{IMB}. 

\smallskip 

\noindent
\textbf{descendent/relative},
also known as the \emph{capped vertex}, is the essential piece in the 
correspondence between descendent and relative insertions. As shown 
in \cite{pcmi}, quantum correction to the capped vertex vanish for any
fixed insertions and sufficiently large framing. This property is 
called \emph{large framing vanishing}. 
Smirnov shows in \cite{S2} how
 to use it to obtain an explicit representation-theoretic 
formula for the capped vertex, which is manifestly a rational function
in all variables. 

\subsubsection{}

The technical crux of the paper is the analysis of the capped 
vertex with our specific insertions $\fb_\alpha$. This is done 
in Section \ref{sp3}. 

Just like the proof of large framing 
vanishing, this is fundamentally a rigidity result in the classical 
spirit of Atiyah, Hirzebruch, Krichever, and others \cites{AtHirz,Kr1,Kr2}. 
 The main 
ingredients in this analysis are the integrality established 
in Section \ref{s_int} and bounds on equivariant weights 
from Section \ref{s_0}.

\subsection{Algebraic Bethe Ansatz reformulation} 
\label{sR1} 

\subsubsection{}
In the study of vertex models of statistical physics, from which 
quantum groups originated, one associates a representation $F$ of 
$\cU_\hbar(\fgh)$ to lines in a 4-valent oriented planar graph and an 
interaction tensor 
$$
R_{F,F'}: F \otimes F' \to F \otimes F' 
$$
to the vertices of the graph, as in Figure \ref{f_r}.  This tensor is the R-matrix for 
$\cU_\hbar(\fgh)$ and the Yang-Baxter equation satisfied by it 
is central to integrability of such models. See e.g.\ 
\cites{ChariPress,JM,KS,Slav} for an introduction. 

\begin{figure}[!hbtp]
  \centering
   \includegraphics[scale=0.33]{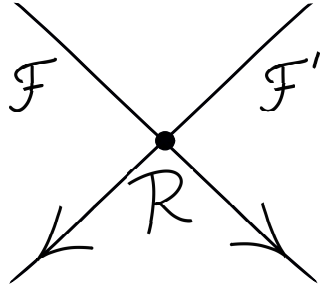}
 \caption{An $R$-matrix interaction in a vertex model}
  \label{f_r}
\end{figure}

In the approach of \cite{MO}, one first constructs geometrically a 
\emph{tensor structure} 
on the K-theory of Nakajima varieties, which then yields R-matrices
and the quantum group itself, see \cites{pcmi,slc} for an
introduction.  

\subsubsection{} 

Tensor structure is realized geometrically using certain 
correspondences called \emph{stable envelopes}, and the 
$R$-matrix is computed as composition of one stable 
envelope with the inverse on another. A certain triangularity
inherent in stable envelopes implies that 
matrix elements of the form 
$$
F \otimes \vac' \xrightarrow{\quad R\big|_{\vac \otimes F'} \quad} \vac \otimes F' 
$$
where $\vac\in F$ and $\vac'\in F'$ are the vacuum, that is, lowest
weight vectors, satisfy
\begin{equation}
R (\alpha \otimes \vac') \Big|_{\vac \otimes F'} = 
\bPi^{-1} \,
\Stab (\alpha \otimes \vac') \Big|_{\vac \otimes F'}\label{RStab}
\end{equation}
for a certain invertible operator $\bPi$ on $F'$. This operator 
$\bPi$ belongs to a very specific commutative algebra
$$
\cB_0 \subset  \End(F') 
$$ 
which may be identified with 
\begin{itemize}
\item[---] the image of the 
quantum loop algebra $\cU_\hbar(\fhh)$ for the Cartan subalgebra 
$\mathfrak{h}\subset \fg$. 
\item[---] the algebra of multiplication operators in the geometric 
realization of $F'$ as a K-theory of a certain algebraic variety. Such
realization makes $F'$ a commutative ring and, in fact, a quotient 
of a ring of $W_G$-invariant Laurent polynomials. It is in this
language that $\bPi$ is presented in \eqref{bPi1} below. 
\item[---] $\cB_0$ is the limit of Baxter's algebra $\cB_z$ of 
commuting transfer matrices 
\eqref{Bax} as the parameter $z$ goes to $0$. 
\end{itemize}
This is reviewed in Section \ref{s_Bax}. 

\subsubsection{}
Our formula for $\bs_\alpha$ is of the form 
\begin{equation}
\bs_\alpha = \Stab(\alpha\otimes \vac') \big|_{\vac \otimes
    \star}\label{s=stab}
\end{equation}
where $\star$ is a specific point \eqref{star}
in the geometric realization of $F'$. Its structure sheaf 
$\cO_\star$ is the unique, up to multiple, eigenvector of 
$\cB_0$ with a certain eigenvalue computed in Section 
\ref{sR2}, where further details may be found. 

This gives 
\begin{equation}
\fb_\alpha = \frac{\bPi} {\Delta_\hbar} \, \cdot \, 
\textup{specific partition function} \label{fbpart}
\end{equation}
where 
\begin{equation}
 \bPi  =\prod_{i\in I} \prod_{k=1}^{\bv_i} \prod_{l=1}^{\bw_i} 
  (1-\hbar x_{i,k}/a_{i,l}) \label{bPi1}
\end{equation}
and the boundary conditions for the partition function in 
\eqref{fbpart} are explained in Figure \ref{f_rrr}.

\begin{figure}[!hbtp]
  \centering
   \includegraphics[scale=0.46]{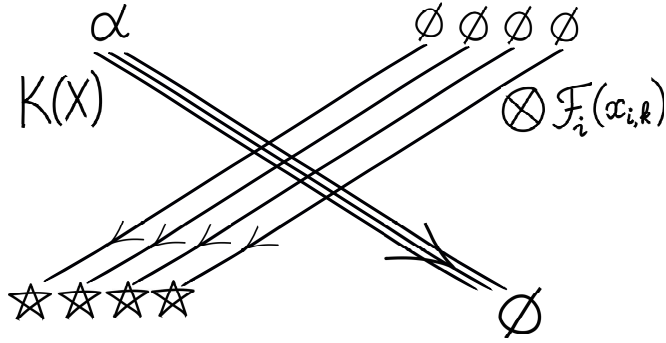}
 \caption{The partition function that computes the function
$\fb_\alpha$. The $\varnothing$-signs denote the vacuum vectors. 
The boundary conditions indicated by stars form an
eigenvector of the algebra $\cB_0$}
  \label{f_rrr}
\end{figure}

 In Figure \ref{f_rrr}, 
we make the fundamental representations $F_i$, $i\in I$, evaluated
at points $x_{i,k}$ where $k=1,\dots, \bv_i$ run along 
the NE-SW lines. Along the NW-SE line runs the 
representation in which $K(X)$ is a weight subspace. 
We draw this line as a multiple line in reference to a tensor 
structure that this module typically possesses. As the boundary 
condition at SW corner, we chose a certain specific eigenvector 
of $\cB_0$. 

The eigenvector property of the boundary conditions means 
the following identity 
\begin{equation}
  \label{fvac}
  \fb_{\alpha\otimes \vac_{ \delta \bw}} = \left. 
\frac{\bPi'}{\bPi\phantom{{}'}}  \right|_{\bw = \delta \bw} \, 
\fb_{\alpha}
\end{equation}
where $\vac_{ \delta \bw}$ is the vacuum vector of weight 
$\delta \bw$ and 
$$
\bPi'  = \prod_{i\in I} \prod_{k=1}^{\bv_i} \prod_{l=1}^{\bw_i} 
  \hbar^{1/2} (1-x_{i,k}/a_{i,l}) \,. 
$$
Pictorially, the eigenvalue property \eqref{fvac} may be 
represented as follows: 
\begin{equation*}
 \raisebox{-25pt}{\includegraphics[scale=0.33]{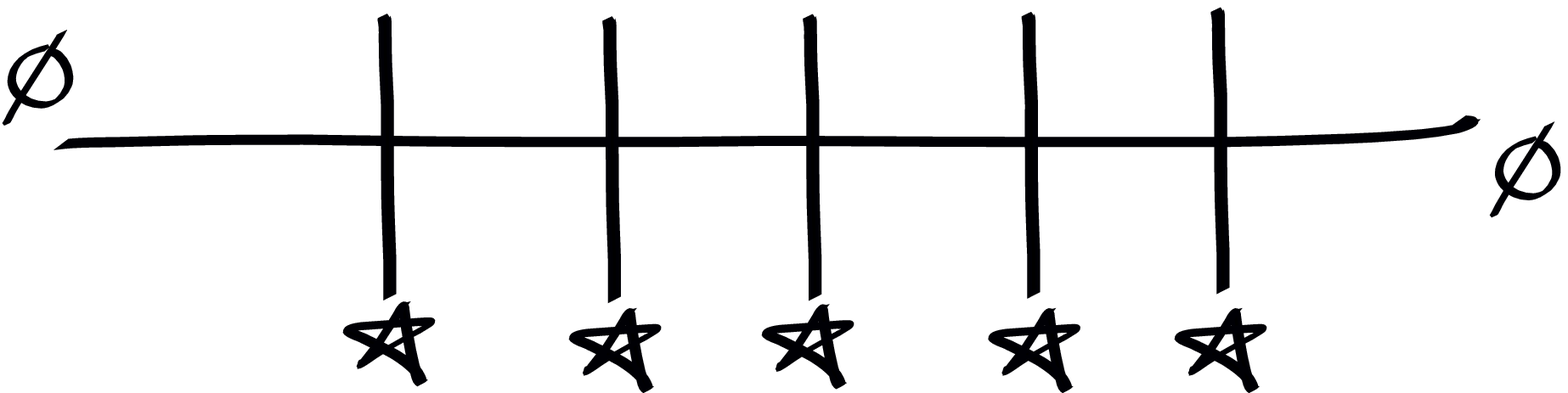}} \quad = 
\left. \frac{\bPi'}{\bPi\phantom{{}'}}  \right|_{\bw = \delta \bw} 
 \raisebox{-25pt}{ \includegraphics[scale=0.33]{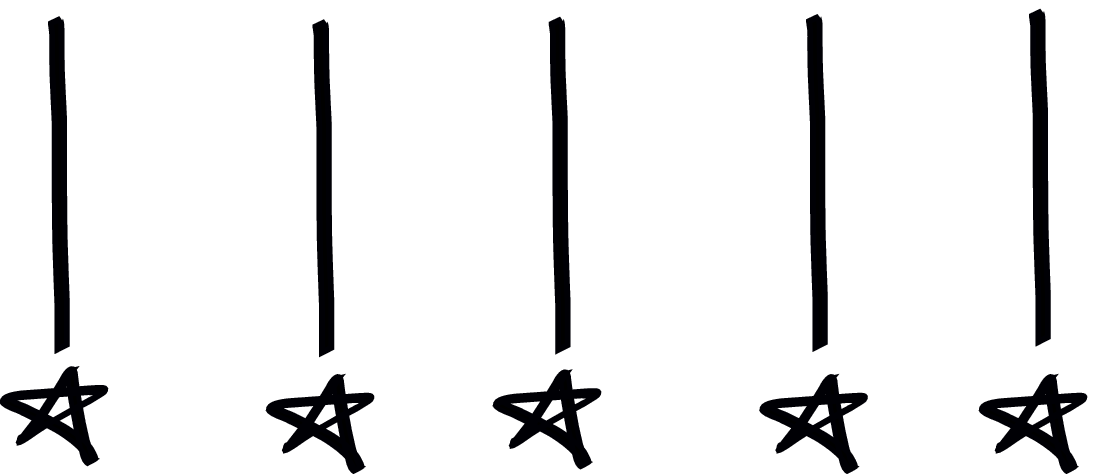}} \label{eq:2}
  \end{equation*}

Explicit formulas for the eigenvector $\cO_\star$ may, in turn, 
be given in terms of stable envelopes. This is a reflection of the 
basic fact that the dual of the stable envelope is again 
a stable envelope with opposite parameters, see \cite{pcmi}.

\subsubsection{}
Let 
$$
\bp(x_{i,k}) \in \Z[T_G/W_G] 
$$
by a symmetric polynomial of $x_{i,k}$, that is, a characteristic 
class $\bp(\{V_i\})$ of the tautological bundles $V_i$ on $\fR$. Formula 
\eqref{fbres} means 
\begin{equation}
  \label{innK}
  \big( \alpha,\bp \big)_{K_\bT(X)} = 
\chi(\alpha \otimes \bp(\{V_i\})) = 
\int_{\gamma_0} \fb_\alpha \, \bp(x_{i,k}) \, \dots \,,
\end{equation}
where $\gamma_0$ is the part of $\gamma$ that encircles 
the image of $\iota$ in \eqref{SpecK} and the integration measure 
omitted in \eqref{innK} 
is, among other things, the specialization of the integration measure
in \eqref{IMB} to quasimaps of degree $0$. 

Using \eqref{innK}, we can read the operators in Figure \ref{f_rrr} 
backwards and interprete that picture as an operator formula
for the off-shell Bethe eigenfunction. In 
the familiar context of the spin 1/2 XXZ spin chain, this 
becomes the classic formula 
$$
\begin{matrix}
    \textup{off-shell Bethe}\\
\textup{eigenfunction} 
  \end{matrix}  = 
B(x_1) \dots B(x_{\bv}) \, \vac
$$
of the algebraic Bethe Ansatz, further generalized in \cite{KR} 
and countless papers since.

\subsection{Acknowledgments}

\subsubsection{} 

Our interactions with Pavel Etingof, Boris Feigin, Edward Frenkel, 
Davesh Maulik, Nikita Nekrasov, 
Nikolai Reshetikhin, and Andrei Smirnov played a very
important role in the development of the ideas presented here.

\subsubsection{}
As already explained, the connection between 
Bethe equations and supersymmetric gauge theories (specifically, enumerative theory of
quasimaps) goes back to the pioneering work of Nekrasov and
Shatashvili \cite{NS1,NS2}. As a next step, 
Bethe eigenvectors found a gauge
theoretic interpretation in Nekrasov's study of 
orbifold defects in gauge theories, see
\cite{NekVid1,NekVid2,NekPrep}.

\subsubsection{} 

This paper looks from a somewhat different angle on 
the problem which Smirnov essentially already solved in \cite{S2} 
building on the large framing vanishing of \cite{pcmi}. Smirnov's 
result is used in \cite{afo} to solve qKZ by Mellin-Barnes integrals 
and the present work was very much motivated by the desire to 
bring the formulas of \cite{afo} 
closer to those of Tarasov and Varchenko. 
In this, we were guided by the papers \cites{RTV1,RTV2} of 
Rim{\'a}nyi, Tarasov, and Varchenko and also by the older papers
of Matsuo \cite{Matsuo} and Reshetikhin \cite{Resh}.

\subsubsection{}

In this paper, we present complete integral solutions to the 
dynamical and qKZ equations for tensor products of evaluation 
representations of quantum affine algebras associated to quivers. 
As a special case, this includes diagonalization of Baxter-Bethe 
commuting operators acting in these spaces. That problem 
goes back to a 1931 paper of Hans Bethe and is the subject of 
an immense body of literature both in mathematics and physics. 

It is unrealistic to analyze how the great many different threads 
present in that literature enter implicitly or explicitly in what we 
do here. We cannot attempt to survey the literature and only 
include those references that influenced our work. Of the many 
different approaches to Bethe Ansatz, we suspect the one based
on the so-called universal weight function \cites{EKP,FKPR,KP,KPT} 
may be the closest. Stable envelopes which we use here give 
a geometric Gauss factorization of the $R$-matrices in the 
style of Khoroshkin and Tolstoy and this is closely related to 
universal weights functions. 

\subsubsection{} 
Another paper which is particularly close to direction of this 
work is \cite{FJMM}, where the authors prove a formula for Bethe 
eigenvectors for $\cU_\hbar(\widehat{\mathfrak{gl}}_1)$ 
which is a close relative of our formula \eqref{fbpart}, 
see Section 4 in \cite{FJMM}. 

Instead of taking the eigenvector boundary condition in 
Figure \ref{f_rrr}, the authors of  \cite{FJMM} take the 
$(\varnothing,\varnothing)$-matrix element of the $R$-matrix 
as a universal map $K_\bT(X) \to \cU_\hbar(\fgh)$, which they 
further compose with a shuffle algebra realization of $\cU_\hbar(\fgh)$ 
to get to functions of $x_{i,k}$. Our formulation bypasses the need to
work with shuffles, and also solves a more general problem --- the 
$q$-difference equations. For eigenvalue problems, overall factors,
such as our denominators $\Delta_\hbar$, are not relevant, which 
explains the discrepancy with \cite{FJMM}, where the square 
$\Delta_\hbar\big|_{\hbar=1}$ of the 
Vandermonde determinant appears in the denominators.

\subsubsection{}
M.A.\ is supported by NSF grant 1521446 and by the Berkeley Center for
Theoretical Physics. Both authors are supported by the 
Simons
Foundation as a Simons Investigators. A.O.\ gratefully acknowledges
funding by the Russian Academic Excellence Project '5-100' and 
RSF grant 16-11-10160.

\section{Main result} 

\subsection{Descendent insertions from stable envelopes} 

\subsubsection{}

For a given oriented framed quiver like the one in Figure \ref{f_quiver}, 
let $\Rep(\bv,\bw)$ 
denote the linear space of quiver representation with dimension 
vectors $\bv$ and $\bw$, where 
$$
\bv_i = \dim V_i \,, \quad \bw_i = \dim W_i \,. 
$$
Let 
$$
\mu: T^*\Rep(\bv,\bw) \to \Lie(G)^*\,, \quad G = \prod GL(V_i) \,, 
$$
be the algebraic moment map and let 
$$
\cZ(\bv,\bw) = \mu^{-1}(0) 
$$
denote its zero locus. 

\begin{figure}[!hbtp]
  \centering
   \includegraphics[scale=0.4]{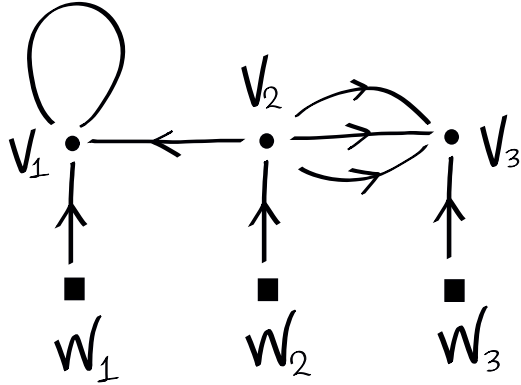}
 \caption{An example of an oriented framed quiver.}
  \label{f_quiver}
\end{figure}

By definition, a Nakajima 
variety $X$ is an algebraic symplectic 
reduction 
$$
X = \cM(\bv,\bw) =  \cZ(\bv,\bw) \rdd G 
$$
where a certain choice of a GIT stability condition is
understood, see e.g.\ \cite{GinzNak} for an 
introduction. The stability choices are parametrized by vectors 
\begin{equation}
\theta \in \R^I = \textup{characters}(G) \otimes_\Z \R \,,
\label{theta}
\end{equation}
which must avoid a finite number of rational hyperplanes, 
up to a positive proportionality. We also consider 
quotient stacks
\begin{equation}
X\subset \fX = 
\left[
\frac{\cZ(\bv,\bw)}{G}
\right] \subset 
\fR = 
\left[
\frac{T^*\Rep(\bv,\bw)}{G}
\right] 
\label{XR}
\end{equation}
obtained by forgetting the stability condition and the 
moment map equations, respectively. 

\subsubsection{}

Our goal in this section is to construct a 
certain $K_\bT(\pt)$-linear map 
\begin{equation}
  \label{bs}
K_\bT(X) \owns \alpha \mapsto \bs_\alpha \in K_\bT(\fR) 
\end{equation}
such that $\bs_\alpha$ is supported 
on $\fX \subset \fR$ and 
\begin{equation}
\bs_\alpha\Big|_{\textup{neighborhood of $X$}}  = \iota_{X,*} \, \alpha  \,,
\label{bsres}  
\end{equation}
where $\iota_{X}: X \hookrightarrow \fR$ is the inclusion. 
One can thus view \eqref{bs} as 
an extension of $\iota_{X,*} \, \alpha$ 
to a K-theory class on $\fR$. 

This extension is canonical once certain further choices are made. 
Its construction involves stable envelopes on a 
larger Nakajima variety $\cM(\bv,\bw+\bv)$.

\subsubsection{}\label{s_orient} 

The dependence of what follows on the stability condition
\eqref{theta} may be summarized as follows. Let $i\in I$ be a 
vertex of the quiver and let $\delta_i\in \N^I$ be the 
delta-function at $i$. Consider 
$$
T^* \Rep(\delta_i,\delta_i) = T^* \Hom(W_i,V_i) \oplus 
T^* \Hom(V_i,V_i)^{\mathbf{g}}
$$
where $\dim W_i = \dim V_i=1$ and 
$$
\mathbf{g} = \textup{number of loops at i}  \,.
$$
The moment map 
equations take the form 
$$
ab = 0\,, \quad  a\in \Hom (W_i, V_i) \,, 
 b\in \Hom (V_i, W_i) \,, 
$$
and 
\begin{equation}
  \label{abstab}
  \cZ(\delta_i,\delta_i)_\textup{stable} = 
  \begin{cases}
    a\ne 0 \,, \quad  \\
b\ne 0 \,, \quad & \textup{depending on $\theta_i \gtrless 0$} \,.
  \end{cases}
\end{equation}
For either choice of stability, this gives 
\begin{equation}
\cM(\delta_i,\delta_i) \cong \C^{2 \mathbf{g}}\label{Mdd}
\end{equation}
equivariantly with respect to 
$
Sp(2 \mathbf{g}) \subset \Aut \cM(\bv,\bw) 
$\,.

Since Nakajima varieties are unchanged under flips of edge 
orientation, we may assume that the direction of the 
invertible map in \eqref{abstab} coincides with the orientation. 
To simplify the exposition we will 
\emph{assume} that the framing edges are oriented 
in the direction of $\Hom(W_i,V_i)$ in Figure
\ref{f_quiver}. It will be clear how to modify 
this in the general case.

\subsubsection{}

By convention, the group $\Ct_\hbar$ scales the cotangent 
direction of $T^*\Rep(\bv,\bw)$ by $\hbar^{-1}$ and thus scales the 
canonical symplectic form $\omega$ 
on $X$ with weight $\hbar$. Note that 
this splitting of the exact sequence 
$$
1 \to \Aut(X,\omega) \to \Aut(X) \to GL(\C\omega) \to 1 
$$
depends on the choice of the orientation. 

\subsubsection{}

Let $V'_i$ be collection of vector spaces of $\dim V'_i =
\bv_i$ and denote
$$
G' = \prod GL(V'_i) \cong G \,. 
$$
Define 
\begin{equation}
Y = \cZ(\bv,\bw+\bv)_\textup{iso} / G'\label{defY}
\end{equation}
where the framing spaces are of the form $W_i\oplus V'_i$ and 
the subscript refers to the locus of points where the framing 
maps 
$$
V_i' \to V_i\,, \quad  \textup{respectively} \quad V_i \to V_i' \,, 
$$ 
are isomorphism, according to the orientation explained in 
Section \ref{s_orient}. In what follows, we will assume that 
$V'_i \xrightarrow{\sim} V_i$.  Clearly,  
$$
 \cZ(\bv,\bw+\bv)_\textup{iso} \subset  
\cZ(\bv,\bw+\bv)_\textup{$G$-stable} \,. 
$$

\subsubsection{}

There is a $G$-equivariant map 
$$
\bdiota :  T^*\Rep(\bv,\bw) \hookrightarrow Y 
$$
which supplements quiver maps by 
\begin{equation}
(\phi, - \phi^{-1} \circ \mu_{GL(V_i)})  \in \Hom(V_i',V_i) \oplus
\Hom(V_i,V_i') \label{extmu}
\end{equation}
for a framing isomorphism  
$$
V_i'\xrightarrow{\,\, \phi\,\,} V_i \,.
$$
The dependence on $\phi$ 
is precisely taken out by the quotient by $G'$. 

We denote the induced map 
\begin{equation}
  \label{bdiota}
  \bdiota :  \fR \hookrightarrow 
\left[Y/G \right]= 
\left[\cM(\bv,\bv+\bw)_\textup{iso}/G' \right] 
\end{equation}
by the same symbol. Formula \eqref{extmu} implies
$$
\mu_{G'} = - \phi^{-1} \circ \mu_G \circ \phi 
$$
and thus $\fX \subset \fR$ 
is cut out but pullback via $\bdiota$ of 
the moment map equations for $G'$. In other words, 
we have a pull-back diagram 
\begin{equation}
  \label{pullbackdiag}
  \xymatrix{
\fX  \ar[r]\ar[d]& \fR  \ar[d]^{\mu_{G'} \circ \bdiota} \\
\left[0/G'\right] \ar[r] &  \left[\Lie(G')^*/G'\right]\,\, .}
\end{equation}

\subsubsection{}\label{s_ZG'} 
Let 
$$
\bU \cong \Ct \subset \cent(G') 
$$
be the group acting with its defining weight $u$ on 
each $V'_i$. We have 
\begin{equation}
  \label{Xfix}
 X \sqcup \cM(\bv,\bv) \subset \cM(\bv, \bw + \bv)^\bU
\end{equation}
and we can choose attracting directions for $\bU$ so that 
$\cM(\bv,\bv)$ lies in the full attracting set of $X$. 

We apply the general machinery of stable envelopes, 
an introduction to 
which may be found in \cite{pcmi},
to this action of  $\bU$. Since $\bU$ commutes with 
$\bT\times G'$, stable envelopes give 
 a $K_{\bT\times G'}(X)$-linear map 
\begin{equation}
  \label{Stab1}
  \Stab: K_{\bT\times G'}(X) \to K_{\bT\times G'}(\cM(\bv, \bw + \bv)) \end{equation}
that depends on two pieces of additional data, 
namely: 
\begin{itemize}
\item[---] a fractional line bundle $\cL \in \Pic(X) \otimes \R$, called
  the \emph{slope}. The slope $\cL$ should be away from the walls of a
  certain periodic locally finite rational 
hyperplane arrangement in $\Pic(X) \otimes \R$ and stable envelopes
depend only on the alcove of that arrangement that contains $\cL$. 
We fix the slope to be 
\begin{equation}
  \label{slope_choice}
  \cL = \varepsilon \cdot \textup{ample bundle} \,, \quad 
0 < \varepsilon \ll 1 \,. 
\end{equation}
This choice is not material for showing \eqref{bsres}, but will 
be crucial for what comes later.  
\item[---] a polarization $T^{1/2}$ which is a solution of the
  equation 
$$
T^{1/2} + \hbar^{-1} \, \left(T^{1/2} \right)^\vee = 
\textup{tangent bundle} 
$$
in equivariant K-theory. Polarization is an 
auxiliary piece of data in that 
stable envelopes
corresponding to different polarizations differ by a shift of the
slope. Polarization is also required to set up quasimap counts, 
see Section 6.1 in \cite{pcmi}, and so we assume that a 
polarization of X has been chosen and set 
\begin{equation}
  \label{choice_pol}
  T^{1/2}_{\cM(\bv, \bw + \bv)} = T^{1/2} X + \sum_i \hbar^{-1}
  \, \Hom(V_i,V'_i) \,,
\end{equation}
that is, we select the directions \emph{opposite}
 to the framing maps 
$V_i'\to V$ that are assumed to be invertible. 
\end{itemize}

\subsubsection{}\label{s_d1}  
Because $\Stab(\alpha)$ is $G'$-equivariant, it descends to 
a class on $Y/G$. We make the following 

\begin{Definition}\label{d1} 
We set
\begin{equation}
  \label{saconstr}
  \bs_\alpha = \bdiota^* \, \Stab(\alpha)  \in K_\bT(\fR) 
\end{equation}
where the slope of the stable envelope is chosen as
 in \eqref{slope_choice}
and the polarization is as in \eqref{choice_pol}. 
\end{Definition}

\begin{Proposition}\label{p1} 
 The class \eqref{saconstr} is supported on $\fX \subset \fR$ and 
satisfies \eqref{bsres}. 
\end{Proposition}

\begin{proof}
  The moment map $\mu_{G'}$ for the group $G'$ is an 
$\bU$-invariant\footnote{that is, an equivariant 
map to a variety with a trivial $\bU$-action} map to an affine variety. Since 
$X$ in \eqref{Xfix} lies in the zero fiber of this map, the 
full attracting set of $X$ does, too. From \eqref{pullbackdiag}, 
we conclude that  
$$
\supp \bs_\alpha \subset \fX\,.
$$
Now let $\Attr(X)$ denote the attracting manifold of $X$ in 
\eqref{Xfix}. It fits into the diagram
\begin{equation}
  \label{Attr}
  \xymatrix{
& \Attr(X) \ar[dl]_{\pi_{\Attr}} \ar[dr]^{\iota_{\Attr}} \\ 
X && \cM(\bv, \bw + \bv)
}
\end{equation}
in which 
\begin{itemize}
\item[---] the map $\pi_{\Attr}$ forgets the maps $V'_i \to V_i$,
\item[---] the map $\iota_{\Attr}$ sets to zero the maps $V_i \to
  V'_i$. 
\end{itemize}

Our choice of the polarization \eqref{choice_pol} and 
the conventions for the normalization of the stable envelope 
explained in Section 9.1 of \cite{pcmi} imply 
$$
\Stab(\alpha) \Big|_\textup{neighborhood of $X$}= \iota_{\Attr,*} \, 
\pi_{\Attr}^* \, \alpha  \,, 
$$
whence the conclusion. 
\end{proof}

\subsection{Restriction to the origin}\label{s_0} 

\subsubsection{} 

As a polynomial in universal bundles, 
the insertion $\bs_\alpha$ is
determined by its restriction to the origin $0\in \fR$. 

Our next goal is to bound the $G$-weights that appear 
in this restriction.  The origin is a fixed point of $G$ and under 
the inclusion $\bdiota$ it corresponds to the point 
\begin{equation}
\star=\{V'_i \xrightarrow{\quad \sim \quad} V_i\,, 
\,\, \textup{all other maps}=0 \} \label{star} 
\end{equation}
which can be viewed as a point in either $Y^G$ or 
$\cM(\bv,\bw+\bv)^{G'}$, the isomorphism in \eqref{star} 
giving an identification of $G$ and $G'$.

To bound the $G$-weight in $\bs_\alpha\big|_0$ is thus 
same as to bound the $G'$ weights of $\Stab(\alpha)\big|_\star$. 
This is equivalent to bounding the $\bA'$-weights, where 
$\bA'\subset G'$ is a maximal torus. 

\subsubsection{}

The torus $\bA'$ contains $\bU$. Since
$X\subset \cM(\bv,\bw+\bv)$ is fixed by the whole torus 
$\bA'$, the triangle lemma for stable envelopes implies 
$$
\Stab_{\bA'}(\alpha) = \Stab_{\bU}(\alpha) 
$$
for the same slope, polarization, and a small perturbation 
of the 1-parameter subgroup. See Section 9.2 in \cite{pcmi} for a 
discussion of the triangle lemma. 

\subsubsection{}

By definition of stable envelopes, the torus weights in their 
restriction to fixed points are bounded in terms of the 
polarization, after a shift by the slope 
\begin{equation}
\cL \in \Pic(X) \otimes_\Z \R = \textup{characters}(G) 
\otimes_\Z \R \label{picdet} \,.
\end{equation}
The identification in \eqref{picdet} sees $\det V_i$ as a 
line bundle on $X$ and as a character of $G$. 

While we made a specific choice of $\cL$ in \eqref{slope_choice}, 
the following proposition is true for an arbitrary slope. 

\begin{Proposition}\label{p_weights} 
The $G$-weights of the restriction of $\bs_\alpha$ to the origin
$0\in\fR$ are contained in 
\begin{equation}
\cL + 
\textup{convex hull} \left( \textup{weights of } \Ld 
\left(T^{1/2}_{\cM(\bv,\bw+\bv)}\right)^{\!\vee}_\star \right) 
\subset \textup{weights} \otimes_\Z \R \,.\label{supps}
\end{equation}
\end{Proposition}

\noindent 
Here vee and star denote the dual representation and 
the restriction to \eqref{star}, that is, to 
$V'=V$, respectively.  

\begin{proof}
Follows directly from the definition of stable envelopes. 
\end{proof}


\subsubsection{}\label{s_actual} 

In general, a polarization of a Nakajima variety 
is a virtual bundle on the prequotient in which either the 
tangent bundle to $G$-orbits or the target of the 
moment map equations enters with the minus sign. 
Note, however, that this term is precisely added back 
in \eqref{choice_pol} after the specialization to $V'=V$.
 This means that
$T^{1/2}_{\cM(\bv,\bw+\bv)}\big|_\star$
is an actual representation of $G$ modulo balanced 
classes, and thus the  exterior algebra in \eqref{supps} is an actual 
$G$-module. 

Recall from \cite{pcmi} that a virtual representation of 
$G$ is called \emph{balanced} if it is of the form $V - V^\vee$, for
some $V\in K_G(\pt)$.  For balanced classes, one defines 
$$
\Ld \left(V - V^\vee\right) = (-1)^{\dim V} \det V \,. 
$$

\subsubsection{}
The bound in Proposition \eqref{p_weights} means that 
$\bs_\alpha$ may be seen as a stable envelope extension 
of the class $\iota_{X,*} \, \alpha$ to the stack $\fR$ in the 
sense of D.~Halpern-Leistner and his collaborators, see \cites{HL,HLMO,HLS}.

\subsection{Integrality of $\fb_\alpha$-insertions}\label{s_int}

\subsubsection{}
Let $\QM(\fR)$ denote the moduli space of stable quasimaps 
$$
f: C \dasharrow \fR \,, 
$$
as defined in  \cite{CKM}, see also e.g.\ \cite{pcmi} for an 
informal introduction. 

By definition, a point of $\QM(\fR)$ is a collection of vector 
bundles $\cV_i$ on $C$ of rank $\bv$, together with a section of 
the associated bundles like $\Hom(\cV_i,\cV_j)$ or 
$\Hom(\cV_i,\cW_j)$ per every 
arrow in the doubled quiver, where $\cW_j$ is a trivial 
bundle of rank $\bw_j$. A quasimap is stable if it evaluates to 
a stable point of $\fR$ at the generic point of $C$.  We set
$$
\deg f = \big( \dots, \deg \cV_i, \dots\big) \in \Z^I  
$$
by definition. 

The image of the natural inclusion 
$$
\iota_{\QM(X)}: \QM(X) \hookrightarrow \QM(\fR)
$$
 is cut out by the 
moment map equations imposed pointwise. 

\subsubsection{}

Consider the pull-back 
$$
\ev_0^* \, \bs_\alpha \in K_{\bT\times \Ct_q}(\QM(\fR))
$$
of the class $\bs_\alpha$ under the evaluation map 
$$
\ev_0: \QM(\fR) \owns f \mapsto f(0) \in \fR \,. 
$$
By Proposition \ref{p1}, every quasimap in the support of this class
satisfies $f(0)\in \fX$. Therefore, the obstruction theory for 
$\QM(X)$, restricted to the support of $\ev_0^* \, \bs_\alpha$ has 
a trivial factor 
\begin{equation}
\Obs_{\QM(X)} \Big|_{\supp \ev_0^* \, \bs_\alpha} 
\to  \hbar^{-1} \, \bigoplus \Hom(\cV_i\big|_0, \cV_i\big|_0)
\to 0 \,, \label{triv_factor}
\end{equation}
corresponding to the moment map equations at $0\in C$. 
We can take the kernel of \eqref{triv_factor} as a new reduced 
obstruction 
theory for $\QM(X)$ to produce a reduced virtual 
fundamental class 
$$
\cO^\vir_{\QM(X),\textup{reduced}}  \in 
K_{\bT\times \Ct_q}(\QM(\fR)) \,. 
$$

\subsubsection{}
The difference between the virtual fundamental class 
$\cO^\vir_{\QM(X)}$ and its reduced version is a factor of 
$$
\Delta_\hbar = \textup{Koszul complex of } 
\left(\bigoplus \hbar^{-1} \Hom(V_i,V_i) \right)\,. 
$$
We make the following 

\begin{Definition} \label{d2} We set 
  \begin{equation}
\fb_\alpha = \Delta_\hbar^{-1} \, \bs_\alpha\label{deffb}
\end{equation}
and we  define the product 
\begin{equation}
\ev_0^*(\fb_\alpha) \otimes \cO^\vir_{\QM(X)}  
\in K_{\bT\times \Ct_q} (\QM(X)) \label{evprod} 
\end{equation}
by the equality of K-classes 
\begin{equation}
\iota_{\QM(X),*} \, \left(\ev_0^*(\fb_\alpha) \otimes \cO^\vir_{\QM(X)}  \right)
\,\, = \,\, 
\ev_0^*(\bs_\alpha) \otimes \cO^\vir_{\QM(X),\textup{reduced}}  \label{evf} 
\end{equation}
on $\QM(\fR)$. 
\end{Definition}

In actual quasimap counting, one uses the so-called 
symmetrized virtual structure sheaves $\tO_\vir$, see Section 6.1 in
\cite{pcmi} and \eqref{tOvir} below. 
Those differ from 
$\cO^\vir$ by a twist by a line bundle, which is the same line bundle
on both sides in \eqref{evf}.

\subsubsection{}
The following is clear from construction

\begin{Proposition}\label{p2} The class \eqref{evprod} is an 
integral K-theory class which equals $\alpha$ for 
$$
\QM(X)_{\textup{degree}=0} \cong X \,. 
$$
Integral formulas for descendent insertions generalize 
verbatim to \eqref{evprod} with the insertion of the rational 
function \eqref{deffb}. 
\end{Proposition}

\subsection{Equivalence of descendent and relative insertions}

\subsubsection{}\label{s_t1} 

Our next goal is to prove the following 

\begin{Theorem}\label{t1}
A relative insertion of $\alpha\in K_\bT(X)$ at $0\in C$ equals 
the descendent insertion of $\fb_\alpha$ at the same point 
in equivariant quasimap 
counts with arbitrary insertions at points away from $0\in C$.  
\end{Theorem}

\noindent 
In Theorem \ref{t1} a certain alignement between the polarization 
used to define $\fb_\alpha$ and a polarization required in 
settings up the
 quasimap counts is understood. Recall that a polarization of 
$T^{1/2}$ of $X$ induces a virtual bundle $\cT^{1/2}$ on the 
domain of the quasimap and one defines the symmetrized
virtual structure sheaf by 
\begin{equation}
  \label{tOvir}
\tO_\vir = \cO_\vir \otimes \left( 
\cK_\vir \otimes \frac{\det \cT^{1/2}_\infty}
{\det \cT^{1/2}_0} \right)^{1/2} \,, 
\end{equation}
where the subscripts denote the fibers of $\cT$ at 
$0,\infty \in C$. 
The quasimap counts from \cite{pcmi} are defined using 
\eqref{tOvir}. Note that they depend on the polarization only 
via its determinant. 

In Theorem \ref{t1}, we assume that the determinants of the 
two polarizations are \emph{inverse} of each other, up to 
equivariant constants. In \cite{pcmi}, equivariant correspondences
are interpreted as operators from the fiber at $\infty$ to the 
fiber at $0$, which is why it is natural to use dual bases
for the fiber at $0$. Stable envelopes, 
in particular, change both the slope and polarization  to opposite
$$
(\cL,T^{1/2}) \mapsto (-\cL, \hbar^{-1} \left( T^{1/2}\right)^\vee 
) 
$$
under 
duality. It is easier to implement the flipping of the polarization in 
the statement of Theorem \ref{t1} than to work with the opposite 
polarization throughout the paper. 

With this change, the localization contributions at $0$ take
the form 
\begin{equation}
\tO_\vir = 
\frac{1}{\Ld_- \, \left(\cT^{1/2}_0\right)^{\!\vee}} \otimes \dots \label{tOvirloc}
\end{equation}
where the dots stand for terms with a  finite limit 
as $q^{\pm 1}\to \infty$ and 
\begin{equation}
\Ld_{-} = \sum_k (-1)^k \Lambda^k  \label{Ld-} \,. 
\end{equation}
See  Section 7.3 of \cite{pcmi} for details on the localization 
formula \eqref{tOvirloc}. 

\subsubsection{} 

The proof of Theorem \ref{t1} proceeds in several steps. 

As a first step, we can 
equivariantly degenerate $C$ to a union 
$$
C \leadsto C_1 \cup_{\textup{node}} C_2
$$
so that $0\in C_1\setminus \{\textup{node}\}$ and all 
other insertions lie in $C_2\setminus \{\textup{node}\}$. 
By the degeneration formula, it is therefore enough to show that
the counts in Theorem \ref{t1} coincide when we impose 
a relative insertion $\beta \in K_\bT(X)$ at $\infty\in C \cong
\bP^1$. 

\subsubsection{}

Since the count of quasimaps relative $0,\infty\in \bP^1$ is the
 glue matrix $\bfG$, the Theorem is equivalent to showing that 
the operator 
\begin{equation}
\alpha \mapsto \tev_{\infty,*} 
\left(\ev_0^*(\fb_\alpha) \otimes \tO_\vir \, z^{\deg} \right)
\in K_{\bT}(X)[q^{\pm 1}][[z]] \label{operglue}
\end{equation}
equals $\bfG$, where $\tev$ is the relative evaluation map 
as in \eqref{eq9}. Here we get polynomials in $q$ because 
the map $\tev_{\infty}$ is proper and $\Ct_q$-invariant. 

Recall that the glue matrix does not depend on $q$, which can be 
explicitly seen by its analysis as $q^{\pm 1} \to 0$ as in 
Section 7.1 of \cite{pcmi}. This analysis is based on 
$\Ct_q$-equivariant localization and we can apply the 
same reasoning to \eqref{operglue}.  

The $\Ct_q$-fixed quasimaps are constant on $\bP^1\setminus 
\{0,\infty\}$ and the contributions from $0$ and $\infty$ 
essentially decouple. The contributions from $\infty$ are 
literally the same as for the glue matrix. They are 
computed using the push-pull in the following 
diagram 
\begin{equation}
  \label{ppinf}
  \xymatrix{ & 
K\left(\QM(X)^{\Ct_q}_{\textup{nonsing at $0$, relative
        $\infty$}}\right) \ar[dr]_{\tev_{\infty,*}} 
\\ 
K(X) \ar[ur]_{\ev_0^*} && K(X) \,\,,}
\end{equation}
where we tensor with $\tO_\vir \, z^{\deg}$ on the middle 
stage. The  analysis in Section 7.1 of \cite{pcmi} shows 
\begin{equation}
  \label{contr_inf}
    \textup{operator from \eqref{ppinf}}
\to 
\begin{cases}
\bfG\,, & q\to 0 \,,   \\
1 \,, & q\to \infty \,. 
\end{cases}
\end{equation}

The contributions from $0\in C$ in the localization formula 
for are 
computed using a parallel push-pull diagram 
\begin{equation}
  \label{pp0}
  \xymatrix{ & 
K\left(\QM(X)^{\Ct_q}_{\textup{nonsing at $\infty$}}\right) 
\ar[dr]_{\ev_{\infty,*}} \\ 
K(\fX) \ar[ur]_{\ev_0^*} && K(X) \,\,,}
\end{equation}
which computes the so-called vertex with descendents, 
see Section 7.2 in \cite{pcmi}. The 
$\Ct_q$-fixed locus in \eqref{pp0} has a concrete 
description as a certain space of flags of quiver representations, see
Section 7.2 of \cite{pcmi} and also \cite{afo}. 

Since \eqref{operglue} is a product of the two operators, 
Theorem \ref{t1} follows from the following 

\begin{Proposition}\label{p3} The vertex with descendent 
$\fb_\alpha$ remains bounded in the $q\to \infty$ limit and 
goes to $\alpha$ in the limit $q\to 0$. 
\end{Proposition}

Note that a vertex with descendents is a power series 
in $z$ and Proposition \ref{p3} implies all terms of 
nonzero degree in $z$ in that series vanish in the $q\to 0$ 
limit.

\subsubsection{}\label{sp3} 
\begin{proof}[Proof of Proposition \ref{p3}]

On the fixed locus, 
the bundles $\cV_i$ can be written in the form 
$$
\cV_i = \oplus_j \cO_C(d_{i,j}[0]) 
$$
with their natural linearization. This means that their fiber 
$\cV_i\big|_\infty$ at infinity is a trivial $\Ct_q$-module 
while the $\Ct_q$-weights in the fiber $\cV_i\big|_0$ at zero 
are $\{q^{d_{i,j}}\}$. This means that the insertion 
$\ev_0^* \bs_\alpha$ is a Laurent polynomial in $\{q^{d_{i,j}}\}$ 
with coefficients in K-theory of the fixed locus. 

This polynomial does not depend on the quiver maps and 
therefore we may assume that all quiver maps are zero. 
The Newton polygon of $\ev_0^* \bs_\alpha$  is thus 
bounded by the formula in Proposition 
\ref{p_weights}. We find 
\begin{equation}
\textup{$\Ct_q$-weights of $\ev_0^* \bs_\alpha$} 
\subset 
(d,\cL) +
\textup{conv} \left( \textup{weights of } \Ld 
\left(\cT^{1/2}_{\cM(\bv,\bw+\bv)}\right)^{\!\vee}_0 \right) 
\label{degq}
\end{equation}
and this inclusion is strict if $(d,\cL)\ne 0$ because it is true 
for an open set of $\cL$. 

Here $\cT^{1/2}_{\cM(\bv,\bw+\bv)}$ is the virtual bundle on $C$ obtained by plugging 
the bundles $\cV_i$ and $\cW_j$ into the formula 
\eqref{choice_pol} and subscript refers to its fiber at $0\in C$. 
As observed in Section \ref{s_actual}, the exterior 
algebra here is a well-defined $\Ct_q$-module. 
Also in \eqref{degq} we have the natural pairing of the 
degree of the quasimap 
$$
d = (d_i) \in \Z^I = H_2(X,\Z) \,, \quad d_i = \sum_j d_{i,j} 
$$
with a fractional bundle $\cL \in \Pic(X) \otimes \R$. 
The moduli spaces of quasimaps of degree $d$ are 
empty unless $d$ is effective, see Section 7.2 in 
\cite{pcmi}, so we assume that $d$ 
is effective in what follows. Since 
$\cL$ was assumed to be an ample bundle, we have
$$
(d,\cL) = 0  \Leftrightarrow d = 0 \,.
$$

{}From \eqref{choice_pol}, we have 
\begin{equation}
\Ld_- 
\left(\cT^{1/2}_{\cM(\bv,\bw)}\right)^{\!\vee}_0 = 
\frac{1}{\Delta_\hbar} 
\Ld_-
\left(\cT^{1/2}_{\cM(\bv,\bw+\bv)}\right)^{\!\vee}_0  \label{LL}
\end{equation}
and therefore from \eqref{tOvirloc} we conclude 
$$
q^{-(d,\cL)} \ev_0^* \fb_\alpha \otimes \tO_\vir  \to 0 \,, 
\quad q\to 0,\infty \,, \quad d\ne 0 \,. 
$$
Since $\cL$ was assumed to be a very small ample bundle, we 
have 
$$
0 < (d,\cL) \ll 1 \,, \quad d\ne 0 \,. 
$$
Therefore for $d\ne 0$ we have 
$$
\ev_0^* \fb_\alpha \otimes \tO_\vir  \to 
\begin{cases}
0 \,, \quad & q\to 0 \,, \\
\textup{bounded} \,, \quad & q\to \infty \,,
\end{cases}
$$
as was to be shown. 
\end{proof}

\section{Reformulations and examples}

\subsection{R-matrices and Bethe eigenfunctions} 
\label{sR2} 

\subsubsection{}

In the setup of Section \ref{s_ZG'} consider the R-matrix 
for the action of $\bU$
\begin{equation}
R: \Unstab^{-1} \circ \Stab \in \End K_{\bT \times G'} 
\left(\cM(\bv,\bw+\bv)^{\bU}\right)_\textup{localized}\label{defR}
\end{equation}
where the map $\Unstab$ is defined as in \eqref{Stab1} 
with the same choice of slope and polarization, but the 
opposite choice of the $1$-parameter subgroup. 

Our next goal is to express $\bs_\alpha$ in terms of 
the restriction of $R(\alpha)$ to $\cM(\bv,\bv)$ in \eqref{Xfix} and 
more concretely in term of its restriction to the $G'$-fixed point 
$\star \in \cM(\bv,\bv)$ as in \eqref{star}. Recall that $\bs_\alpha$ is 
completely determined by its restriction to the point 
$\star$. 

\subsubsection{}

By our choice of the $1$-parameter subgroup, $\cM(\bv,\bv)$ was 
at the bottom of the attracting order among components of the 
$\bU$ fixed locus. Since this order is reversed for $\Unstab$, we
have 
\begin{equation}
\Unstab(\beta)\big|_{\cM(\bv,\bv)} = \beta \, \otimes \Ld_{-} 
\left(
N_\textup{repell}^\vee \right) \otimes \dots \,, \label{Unstab_dots}
\end{equation}
for any $\beta \in K_{T\times G'}(\cM(\bv,\bv))$, 
where dots stand for a certain line bundle 
and $N_\textup{repell}$ is the repelling part of the 
of the normal bundle $N$ to $\cM(\bv,\bv)$. 
We have 
\begin{equation}
N = \underbrace{\sum_i \Hom(W_i,V_i)}_{\textup{attracting for
    $\Unstab$}} + \underbrace{\hbar^{-1} \sum_i \Hom(V_i,W_i)}
_{\textup{repelling for
    $\Unstab$}} \,. \label{N}
\end{equation}
Fixing the line bundle in \eqref{Unstab_dots}
 requires fixing a polarization 
of $X$. For simplicity, we \emph{assume}
 that the polarization of the 
framing maps for $X$ is the same as in the new 
framing terms in \eqref{choice_pol}, that is 
\begin{equation}
  \label{T1/2X}
  T^{1/2}X = \hbar^{-1} \sum_i \Hom(V_i,W_i) +  
\textup{non-framing terms}  \,.
\end{equation}
Recall from Section \eqref{s_orient} that such choice 
of orientation on 
framing edges was dependent on the stability parameter
$\theta$, and that both orientation and polarization should be 
flipped if the entries of $\theta$ change sign. 

\subsubsection{}

With the assumption \eqref{T1/2X},
the repelling directions in \eqref{N} 
coincide with the normal directions 
chosen by polarization and hence the dots in \eqref{Unstab_dots}
are trivial. In other words
\begin{equation}
\Unstab(\beta)\big|_{\cM(\bv,\bv)} = \beta \, \otimes \, \bPi \,.\label{Unstab}
\end{equation}
where
\begin{equation}
  \bPi  = \Ld_{-} 
\left(
\hbar \sum_i \Hom(W_i,V_i) \right)  
= \prod_{i\in I} \prod_{k=1}^{\bv_i} \prod_{l=1}^{\bw_i} 
  (1-\hbar x_{i,k}/a_{i,l})  \,. 
\label{bPi}
\end{equation} 
The variables $x_{i,k}$ and $a_{i,l}$ in \eqref{bPi} 
are the Chern roots of $V_i$ and $W_i$, respectively, as 
in \eqref{torus A}. 
We deduce the following 

\begin{Proposition}\label{p4} We have 
  \begin{equation}
    \label{sR}
    \bs_\alpha \big|_0= \bPi \,\, R(\alpha) \big|_\star 
\,.   
  \end{equation}
\end{Proposition}

\subsubsection{}\label{s_Bax} 
It remains to characterize the fiber at $\star$, which is, 
abstractly, a linear form 
\begin{equation}
K_{G'}(\cM(\bv,\bv)) \owns \cF \mapsto
 \cF\big|_\star = \chi(\cF \otimes \cO_\star) \in
 K_{G'}(\pt) \label{fiber_star}\,, 
\end{equation}
in representation-theoretic terms. 

The structure sheaf 
$$
\cO_\star \in K_{G'}(\cM(\bv,\bv))
$$
of the $G'$-fixed point \eqref{star} is an eigenvector of operators 
of multiplication in $K_{G'}(\cM(\bv,\bv))$, namely 
\begin{equation}
\cF \otimes \cO_\star = \cF\big|_\star \cdot \cO_\star \label{mulF}
\end{equation}
for any $\cF\in K_{G'}(\cM(\bv,\bv))$. 

Following \cite{MO}, we recall how express generators of the
commutative algebra of operators \eqref{mulF} in terms of 
the vacuum matrix elements of R-matrices. 
These are operators in $K_{G'}(\cM(\bv,\bv))$ defined by
\begin{equation}
R_{\bw,\varnothing,\varnothing} (\beta) = R(\beta)|_{\cM(\bv,\bv)} 
\,, \label{R00}
\end{equation}
where $R$ is our current R-matrix defined in \eqref{defR}. 
Its dependence on the dimension 
vector $\bw$ is made explicit in \eqref{R00}. Obviously 
\begin{equation}
R_{\bw,\varnothing,\varnothing} = \lim_{z\to 0} 
\tr_\textup{1st factor} (z^{\bv} \otimes 1) R \label{Bax}
\end{equation}
and so the operators \eqref{R00} are the limit of Baxter's commuting 
transfer matrices as $z\to 0$. 

In the description \eqref{N} of the normal bundle, the 
repelling direction for $\Stab$ are the attracting directions 
for $\Unstab$ and they are precisely opposite to the polarization. 
Therefore
$$
\Stab(\beta) = \bPi' \otimes \beta \,, \quad \beta \in 
K_{\bG'}\left(\cM(\bv,\bv)\right) \,, 
$$
where 
\begin{align}
  \bPi'  &= \hbar^{\frac14\rk N} \, \Ld_{-} 
\left(
\sum_i \Hom(W_i,V_i) \right) \notag  \\
&= \prod_{i\in I} \prod_{k=1}^{\bv_i} \prod_{l=1}^{\bw_i} 
  \hbar^{1/2} (1-x_{i,k}/a_{i,l})  \,. 
\label{bPip}
\end{align} 
{}From this and \eqref{Unstab} it follows that 
\begin{equation}
R_{\bw,\varnothing,\varnothing}  = 
\frac{\bPi'}{\bPi\phantom{{}'}} \otimes \textup{---} \,\, 
\in \End 
K_{G'} (\cM(\bv,\bv)) \otimes \Q(\bA') \,. 
\end{equation}

\subsubsection{}

Recall that $\bA'\subset G'$ denote the maximal torus. 
Extending the analysis of Section \ref{s_orient}, 
it is easy to see that 
$$
\cM(\bv,\bv)^{\bA'}_{\textup{component of $\star$}} = 
\prod_i \cM(\delta_i,\delta_i)^{\bv_i}  \,. 
$$
This is a vector space with origin $\star$. The Weyl 
group of $G'$ acts on it by permutations of factors. 
Since the K-theory of this fixed component is trivial, 
we have the following 

\begin{Proposition}
 The structure sheaf $\cO_\star$ is the unique, up to 
multiple, eigenvector of the operators 
$R_{\bw,\varnothing,\varnothing}$ with eigenvalue 
\begin{equation}
\label{eigR00} 
  R_{\bw,\varnothing,\varnothing} (\cO_\star) = 
\left. \frac{\bPi'}{\bPi\phantom{{}'}} \right|_{x_{i,k}=a'_{i,k}} \, 
\cO_\star 
\end{equation}
and \eqref{fiber_star} is the unique, up to multiple, linear form in
the dual of this eigenspace. The normalization may be fixed by 
e.g.\ \eqref{fbres}. 
\end{Proposition}

To connect with the notations of Section \ref{sR1} of 
the Introduction, it suffices to make the inverse 
substitution $a'_{i,k} = x_{i,k}$. 

\subsection{Example: $\cUl$}\label{cUl} 

\subsubsection{}

Our goal here is to produce an explicit basis of the functions 
$\fb_\alpha$ for quivers of cyclic type $\Ah_{\ell-1}$ with $\ell$
vertices. The corresponding Nakajima varieties are moduli spaces
of framed sheaves, including Hilbert schemes of points,
 on the $A_{\ell-1}$-surfaces, that is, minimal 
resolutions of 
$$
x y = z^{\ell} \,,
$$
starting with the affine plane $A_0=\C^2$ for $\ell=1$. 
In particular, K-theoretic counts of 
quasimaps to these Nakajima varieties are 
directly related to K-theoretic Donaldson-Thomas theory of 
threefold fibered in $A_{\ell-1}$-surfaces. 

The Lie algebra $\fg$ corresponding to the cyclic quiver 
is the affine Lie algebra $\fgl_\ell$, hence the action of
a double affine algebra $\cUl$ on the K-theories of these
Nakajima varieties. 
Its direct link to  important questions in enumerative geometry and 
mathematical physics makes $\cUl$ a very interesting object of 
study. See in particular \cite{NeTh} for a detailed discussion and 
many references. 

As a special case, cyclic quiver varieties include quiver 
varieties for the linear quiver $A_\ell$, for which we recover the action of 
$\cU_\hbar(\fgl_{\ell+1})$ and the formulas of Tarasov and Varchenko. 
The connection between those formulas and stable envelopes
has already been observed in \cites{RTV1,RTV2}. 

\subsubsection{}
For explicit formulas, it is convenient to choose a particularly
symmetric polarization of $X$. We start with a polarization 
\begin{equation}
T^{1/2} = \bigoplus_{\circ\,  \to \, \circ'} 
\Hom(\, \circ \, , \, \circ' \,) - 
\bigoplus_{i} 
\Hom(V_i , V_i) 
\label{pol0} 
\end{equation}
obtained from an orientation of 
the framed quiver in Figure  \ref{f_quiver}. The first 
sum in \eqref{pol0} is over all oriented edges. 
The weights in the corresponding stable envelopes are 
then bounded by the weights in $\Ld_- \left(T^{1/2}\right)^\vee$, 
which is an product of expressions like 
\begin{equation}
\Ld_- \Hom(V,V')^\vee = \prod (1-x_i/x'_j)\label{LdHom} \,. 
\end{equation} 
Here $\Ld_-$ is the alternating sum of exterior powers as 
in \eqref{Ld-} and $\{x_i\},\{x'_j\}$ are the Chern roots of 
$V$ and $V'$, respectively.

We define 
\begin{align}
  \Ldi \Hom(V,V') &= \Ld_- \Hom(V,V')  \otimes \left(\det
                    V\right)^{\rk V'} \notag \\
& = \aroof \left(\Hom(V',V)\right) \otimes \left(\det
                    V\right)^{\frac12 \rk V'} \otimes \left(\det
                    V'\right)^{\frac12 \rk V} \notag \\
& =\prod (x_i-x'_j) \label{Ldi}
\end{align}
which is, up to a sign, symmetric in $V$ and $V'$. 
Since \eqref{LdHom} and \eqref{Ldi} differ by a sign and a line 
bunde, we have 
\begin{equation}
\Ldi \, T^{1/2} = \pm \, \Ld_- \left(T^{1/2}_\diamond
\right)^\vee \label{Tdia}
\end{equation}
for a certain polarization $T^{1/2}_\diamond$. In what follows, 
we consider stable envelopes with this polarization; their weights 
are bounded by \eqref{Tdia}. 

It is convenient to extend the definition \eqref{Ldi} by 
linearity in the second factor 
\begin{equation}
  \label{Ldiy}
  \Ldi \left(\Hom(V,V') \otimes M \right)=  \Ldi \Hom(V,V' \otimes M) 
=\prod_{i,j,k} (x_i- m_k x'_j ) 
\end{equation}
where $M$ is a multiplicity bundle and $\{m_k\}$ are its Chern 
roots.  Recall that for Nakajima varieties may have nontrivial 
automorphisms acting on edge multiplicity spaces. The rank 
of the group of such automorphisms is the 1st Betti number 
of the quiver. A review of these basis facts may be found 
e.g.\ in the introductory material in \cite{MO}.

\subsubsection{}
Let $\bA\subset \bT$ denote the subtorus preserving the 
symplectic form $\omega$. The torus $\bA$ 
includes a maximal torus of the framing group $GL(W)$ and 
an additional $\Ct_\textup{loop}$ for the loop in the quiver. 
In the moduli of sheaves interpretation, this $\Ct_\textup{loop}$ 
acts by symplectic automorphisms of the surface. 

We have 
\begin{equation}
X^\bA = \bigcup_{\sum \vt^{(ij)} = \bv} \,\, 
\prod_{i\in I} \prod_{j=1}^{\bw_i} \cM_\textup{linear}(\vt^{(ij)},
\delta_i)\label{XA1}
\end{equation}
where the $\cM_\textup{linear}$ denotes the Nakajima variety 
corresponding to the infinite linear quiver $A_\infty$ and the 
equality $\sum \vt^{(ij)} = \bv$ in \eqref{XA1} involves
summing over the fibers of the map 
\begin{equation}
A_\infty \, \xrightarrow{\quad \textup{universal cover}\quad} \,
\Ah_\ell \label{AiAl} \,. 
\end{equation}
The fixed locus \eqref{XA1} may be interpreted as a Nakajima 
variety associated to a (disconnected) fixed-point quiver 
\begin{equation}
Q^\bA =  \textup{$|\bw|$-many copies of $A_\infty$}  
\label{QbA}
\end{equation}
with dimension vector $\vt = \big(\vt^{(ij)}_k\big)$, 
where $|\bw|= \sum \bw_i$. 

\subsubsection{}

Note that 
\begin{equation}
  \label{vpart}
  \cM_\textup{linear}(\bv,
\delta_i) = 
\begin{cases}
  \pt\,, & \textup{$\bv$ corresponds to a partition $\lambda$} \,, \\
 \varnothing\,, & \textup{otherwise}\,, 
\end{cases}
\end{equation}
where the first case means that 
$$
\bv_j = \textup{\# of squares in $\lambda$ of content $j-i$} \,, 
$$
with
$$
\textup{content}(\square) = \textup{column}(\square) - 
\textup{row}(\square) \,. 
$$
Indeed, the nonempty moduli spaces in \eqref{vpart} form 
a basis of a level one Fock module for $\fgl_\infty$, also known 
as a fundamental representation of this Lie algebra. Those 
are labelled by an integer $i$ and this is the index $i$ in the 
formulas above. 

\subsubsection{}

Let $F$ be a component of the fixed locus \eqref{XA1}. It 
corresponds to a homomorphism 
$$
\phi_F: \bA \to G
$$
which makes all spaces $V_i$ and the $\Hom$-spaces  
between them $\bA$-graded. In particular, the fixed locus $F$ 
itself parameterizes $\bA$-invariant 
quiver maps, modulo the action of the 
centralizer $G^\bA\subset G$. 

We choose a generic $1$-parameter subgroup in $\bA$ to 
partition all nonzero weights into attracting and repelling. 
In particular, the polarization $T^{1/2}$ decomposes 
\begin{equation}
T^{1/2} \Big|_{F} = \left(T^{1/2}\right)_\textup{attracting} \oplus 
\left(T^{1/2}\right)_\textup{$\bA$-fixed} \oplus
\left(T^{1/2}\right)_\textup{repelling}\label{TdecA} 
\end{equation}
according to the $\bA$ weights. We define 
\begin{equation}
\fb_F = \sum_{w\in W_G/W_{G^\bA}} w \cdot \Ldi \left(
\left(T^{1/2}\right)_\textup{repelling} \oplus 
\hbar \left(T^{1/2}\right)_\textup{attracting} \right)\label{fbF} \,, 
\end{equation}
where the Weyl group acts by permuting the Chern roots of the 
bundles. Since the decomposition \eqref{TdecA} is 
$G^\bA$-equivariant, the group the Weyl group $W_{G^\bA}$
of $G^\bA$ 
acts trivially and the summation in \eqref{fbF} is over the 
cosets of $W_{G^\bA}$. 

Let $\cL$ be line bundle of the form 
\begin{equation}
  \label{cLdia}
  \cL_\diamond = \bigotimes \left(\det V_i \right)^{\varepsilon_i} \,, 
\quad 0 < \varepsilon_i \ll 1 \,. 
\end{equation}
The following proposition may be seen as an instance of 
an abelianization formula for stable envelopes, see e.g.\ \cites{Sh,
S1,ese}. Closely related constructions also appear in \cites{HLMO,
HLS}.

\begin{Proposition}\label{pfbF} 
The functions $\fb_F$ for all components $F$ of the fixed locus 
\eqref{XA1} form a $\Q(\bT)$-basis of the space of functions 
$\fb_\alpha$ for cyclic quiver varieties for polarization \eqref{Tdia} 
and slope \eqref{cLdia}\,. 
\end{Proposition}

\noindent 
Note that if the rest of the terms and the cycle of integration in 
\eqref{IMB} are symmetric then there is no need to symmetrize 
under the integral sign. 

\subsubsection{}\label{s_Hilb} 

For example, let
$$
X = \Hilb(\C^2,n)
$$
be the Hilbert scheme of $n$ points in the plane
$\C^2$, which corresponds to 
$$
\ell=1\,, \quad \bw=1\,, \quad \bv=n \,. 
$$
The tori 
$$
\bT = \left\{ 
  \begin{pmatrix}
    t_1 \\ & t_2 
  \end{pmatrix} \right\} 
\supset 
\bA = \left\{ 
  \begin{pmatrix}
    t_1 \\ & t_1^{-1}  
  \end{pmatrix} \right\} 
$$
acts naturally on $\C^2$ and $\Hilb(\C^2,n)$ and 
$$
\hbar = \frac{1}{t_1 t_2} \,. 
$$
The fixed points of $\bT$ and $\bA$ 
are indexed by partitions $\lambda$ of $n$ and 
$$
V\big|_{\lambda} = \sum_{\square=(i,j) \in \lambda} 
t_1^{1-j} t_2^{1-i} 
$$
as a $\bT$-module. In particular, the $\bA$-weights in $V$ are 
given by minus contents of the boxes.  As a polarization, we 
may take 
\begin{align*}
  T^{1/2} & = V + (t_1-1) \Hom(V,V) \\
& = \sum x_i + (t_1-1) \sum_{i,j} x_i/x_j 
\end{align*}
where $\{x_i\}$ are the Chern roots of $V$. A fixed point is 
specified by the assignment of $x_i$ to the boxes of 
$\lambda$, up to permutation. 

If we take $t_1$ to be a \emph{repelling} weight for $\bA$ then 
$$
T^{1/2}_{\gtrless} = \sum_{c(i)\gtrless 0} x_i + 
t_1 \sum_{c(i)\gtrless c(j)+1 } x_i/x_j - \sum_{c(i)\gtrless c(j) } x_i/x_j
$$
where 
$$
T^{1/2}_> = T^{1/2}_\textup{attracting} \,, \quad 
T^{1/2}_< = T^{1/2}_\textup{repelling}\,,
$$
and $c(i)$ is the content of the box in $\lambda$ assigned to $x_i$. 
Therefore, up to an $\hbar$ multiple, we have 
$$
\fb_\lambda = \textup{symmetrization of } \frac{\Pi_1 \Pi_2} {\Pi_3} 
$$
where 
$$
\Pi_1 = \prod_{c(i)<0} (1 - x_i) \prod_{c(i)>0} (t_1 t_2 - x_i) 
$$
and 
\begin{align*}
\Pi_2  &= \prod_{c(i)<c(j)+1}(x_j - t_1 x_i) \prod_{c(i)>c(j)+1}(t_2
        x_j - x_i) \\
\Pi_3  &= \prod_{c(i)<c(j)}(x_j -  x_i) \prod_{c(i)>c(j)}(t_1 t_2
        x_j - x_i) \,. 
\end{align*}
These are formulas for K-theoretic stable envelopes for 
$\Hilb(\C^2,n)$ with the polarization and slope 
as in Proposition \ref{pfbF}. They are a direct 
K-theoretic generalization of the formulas from 
\cites{Sh,S1}. 

Note that in all cases treated by the formula 
\eqref{fbF} the slope is near an integral line bundle. Much 
more interesting functions appear at fractional slopes, but 
they seem to be not required in the context of Bethe 
Ansatz.

\subsubsection{}

The proof of Proposition \ref{pfbF}  takes several steps. As a first 
step, we clarify the geometric meaning of the formula 
\eqref{fbF}.

We separate the numerator and denominator in \eqref{fbF} 
by writing 
$$
\left(T^{1/2}\right)_\textup{repelling} \oplus 
\hbar \left(T^{1/2}\right)_\textup{attracting} 
= \rho_+  - \rho_-
$$
as a difference of two $\bA$-modules. Then 
$\Ldi \rho_+$ is the numerator in \eqref{fbF}, while 
$\Ldi \rho_-$ is the denominator. We note that 
\begin{equation}
\Ldi \rho_+ = \pm \cO_{\Attr(F)} \otimes 
\dots \in K_{\bT\times P} (T^*\Rep) \label{Ldirho}
\end{equation}
where dots stand for a character. Here 
$$
\Attr(F) \subset T^*\Rep
$$
is the $\bA$-attracting manifold and $P\subset G$ is the 
the parabolic subgroup with 
$$
\fp = \Lie P = \fg_\textup{attracting}  \,.
$$
It acts on the character in \eqref{Ldirho} via the homomorphism 
\begin{equation}
  \label{NPG}
  1 \to \textup{unipotent radical $N$} \to P \to G^\bA \to 1
\end{equation}
to its Levi subgroup $G^\bA=P^\bA$. 

\subsubsection{}

Formula \eqref{Ldirho} illustrates two general facts. 
First, this is an instance of stable 
envelopes for abelian quotients and abelian stacks. 
In general, in the abelian case, stable envelopes are structure
sheaves of the attracting locus, up to line bundles.

The second general principle apparent in \eqref{Ldirho} 
 is summarized in the following, in which $Y$ is an abstract 
variety or stack for which stable envelopes are defined. 

\begin{Lemma} Let $P$ in 
$$
\bA \subset P \subset \Aut(Y) 
$$
be an algebraic group such that the $\bA$-weights in $\fp$ are 
attracting. Stable envelopes define a map 
$$
K_{P} (Y^\bA) \to K_P(Y)
$$
where $P$ acts on $Y^\bA$ via the projection to $P^\bA$. 
\end{Lemma}

\begin{proof}
Our assumption on $P$ implies that it preserves attracting 
manifolds. We then 
argue inductively using the attracting order on the components
$F_i$ of $Y^\bA$. For the very bottom component, the stable 
envelope is the push-pull in the $P$-equivariant diagram 
\begin{equation}
  \label{Fbottom}
  \xymatrix{ & 
\Attr\left(F_\textup{bottom}\right)
\ar[dl]_{\textup{projection\quad}}
\ar[dr]^{\textup{\quad inclusion}}
\\ 
F_\textup{bottom} && Y \,,}
\end{equation}
up to a line bundle pulled back from $F_\textup{bottom}$. 
For all other components $F$, stable envelopes are uniquely 
determined by having the same structure \eqref{Fbottom} 
near $F$ and being orthogonal to all lower stable envelopes
in the sense of \cite{HL}, whence the conclusion. 

\end{proof}

\subsubsection{}

By construction 
\begin{equation}
\rho_- = \fg/\fp \, \oplus \, \hbar \, \fn \,, 
\label{rho-}
\end{equation}
where $\fn = \Lie N$ is the nilradical of $\fp$. The second 
term here has the following interpretation. 

Since the moment map is a $\bA$-equivariant map, we have
$$
\mu: \Attr(F) \to \fg^\vee_\textup{attracting} = \fn^\perp \,. 
$$
Therefore there is no need to impose the moment map in $\fn^\vee$. 
Equivalently, if we planning to multiply by the Koszul complex 
$\Delta_\hbar$ of 
$\hbar^{-1} \fg^\vee$ to get a class supported on $\fX$, we may 
divide by 
$$
\textup{Koszul complex of $\hbar^{-1} \fn^\vee$} = 
\pm \Ldi \left(\hbar \, \fn \right) \otimes \dots \,,
$$
where dots stand for an unspecified character, as before. 

\subsubsection{}

The meaning of the first term in \eqref{rho-} is the following. 
Given a $P$-equivariant sheaf on a $G$-variety $Y$, we can induce it to a 
$G$-equivariant by first, making a $G$-equivariant sheaf on 
$G/P \times Y$ and then pushing it forward to $Y$. 

In the case at hand, up to a line bundle, 
the denominators $\Ldi \left(\fg/\fp\right)$ and 
the summation over $W_G/W_{G^\bA}$ in \eqref{fbF} quite 
precisely come from an equivariant localization on $G/P$.  
We conclude the following 

\begin{Proposition}
The symmetric polynomial $\Delta_\hbar \, \fb_F$ 
represents a class in
$K_\bT(\fR)$ supported on the full $\bA$-attracting set of 
$F$ in $\fX$ 
\end{Proposition}

\subsubsection{}

\begin{proof}[Proof of Proposition \ref{pfbF}]

Note that the formula \eqref{fbF} is universal for all dimension 
vectors. By the logic of our Definition \ref{d1}, to prove \eqref{fbF} for 
a specific $X=\cM(\bv,\bw)$ we need to check something 
for a larger Nakajima variety $\cM(\bv,\bw+\bv)$. Namely, together
with $X$, the 
fixed locus $F$ embeds in $\cM(\bv,\bw+\bv)$ and 
we need to bound $\bU$-weights in the corresponding function 
$\fb_{F, \cM(\bv,\bw+\bv)}$. 

{}From this angle, 
there is nothing special about the framing dimension being 
increased by exactly $\bv$, and we can more generally assume 
that an action of $\bU\cong \Ct$ is defined 
by a decomposition of the framing spaces 
$$
W = W' + u \, W'' \,, 
$$
in which $u$ is the defining weight of $\bU$ and $W',W''$ are 
trivial $\bU$-modules. We have 
$$
X^\bU = \bigcup_{\bv'+\bv''= \bv} \cM(\bv',\bw') \times 
\cM(\bv'',\bw'')
$$
and we choose the attracting directions so that 
components with larger $\bv''$ are attracted to those with 
smaller $\bv'''$. For the bundle $\cL_\diamond$ from \eqref{cLdia}
we have 
$$
\textup{weight} \, \cL_\diamond \Big|_{X^\bU} = 
\varepsilon \cdot \bv'' \,. 
$$
In the context of our Definition \ref{d1},  
\begin{itemize}
\item[---] we assume that $F$ lies in the $\bw''=0$ component of 
the fixed locus $X^\bU$, 
\item[---] also assume that the attracting direction for $\bA$ agree
  with those for $\bU\subset \bA$, 
\item[---] and we need to prove that 
  \begin{equation}
    \label{boundU}
  \left. u^{- \varepsilon \cdot \bv'' } \frac{\fb_F}{\Ldi \, T^{1/2}}
    \, \right|_{X^\bU} = O(1) \,, 
\quad u^{\pm 1} \to \infty \,, 
  \end{equation}
for $0<\varepsilon_i \ll 1$. 
\end{itemize}

In \eqref{fbF}, we select the attracting and repelling 
directions in the decomposition \eqref{TdecA}. Since 
in \eqref{boundU} this is compared with the whole polarization, 
the bound \eqref{boundU} follows from 
  \begin{equation}
    \label{boundU2}
  \left. u^{- \varepsilon \cdot \bv'' } \frac{1}{\Ldi \, T^{1/2}_\textup{$\bA$-fixed}}
    \, \right|_{X^\bU} = O(1) \,, 
\quad u^{\pm 1} \to \infty \,, 
  \end{equation}
which will now be established. 

In $\fb_F$, the Chern classes $x_{ij}$ of the universal bundles 
are partitioned into various groups according to their
$\bA$-grading. The sizes of these groups are 
given by the dimension 
vector $\vt$ of the quiver \eqref{QbA}. Restricted to 
$X^\bU$, this dimension vector further splits 
$$
\vt = \vt' + \vt''
$$
into components of weight $0$ or $1$ with respect to $\bU$. 

For computations of degree in $u$, it is natural to use the 
quadratic form associated to the quiver \eqref{QbA}. In general, 
for any quiver $Q$ with dimension vector $\bv$, one defines 
\begin{equation}
(\bv,\bv)_Q = \sum_{i \to j} \bv_i \bv_j \label{formQ}
\end{equation}
where $i\to j$ means that $i$ and $j$ are connected by an edge of
$Q$. Together with the corresponding dot product 
$$
\bv \cdot_Q \bv' = \sum_{i\in \textup{vertices}(Q)} \bv_i \bv'_i 
$$
the form \eqref{formQ} 
enters the dimension formula for Nakajima varieties 
\begin{equation}
  \label{dimNak}
  \tfrac12 \dim \cM_Q(\bv,\bw) = (\bv,\bv)_Q + \bv \cdot_Q (\bw - \bv) 
\,. 
\end{equation}
To prove \eqref{boundU2}, we consider the cases $u\to 0$ and 
$u\to\infty$ limits separately. In the $u\to 0$ limit, we have
  \begin{equation}
    \label{boundU0}
 \left. \frac{1}{\Ldi \, T^{1/2}_\textup{$\bA$-fixed}}
    \, \right|_{X^\bU} = O(u^{e_0}) \,, 
\quad u\to 0 \,, 
  \end{equation}
where 
\begin{equation}
  \label{e0}
  e_0 = - \vt'' \cdot_{Q^\bA} \vt'' + (\vt'',\vt'')_{Q^\bA} 
\end{equation}
because $\bw''=0$ by construction.  Since $Q^\bA$ is a union 
of quivers of type $A_\infty$ the quadratic form in \eqref{e0}, 
which is proportional to the Cartan-Killing form for the corresponding 
Lie algebra, is 
negatively defined. Therefore 
$$
e_0 < 0 \quad \textup{for}\quad \bv''\ne 0 
$$
and the $u\to 0$ case of \eqref{boundU2} is established. 

In the opposite limit we have 
  \begin{equation}
    \label{boundUinf}
 \left. \frac{1}{\Ldi \, T^{1/2}_\textup{$\bA$-fixed}}
    \, \right|_{X^\bU} = O(u^{e_\infty}) \,, 
\quad u\to \infty \,, 
  \end{equation}
where 
\begin{equation}
  \label{einf}
  e_\infty = \tfrac12 \dim \cM_{Q^\bA}(\vt,\bw) - 
\tfrac12 \dim \cM_{Q^\bA}(\vt',\bw) \,. 
\end{equation}
{}From \eqref{vpart}, we conclude 
$$
e_\infty = 0 
$$
and the proof of \eqref{boundU2} is complete. 
\end{proof}

 \renewcommand{\theequation}{A.\arabic{equation}}
  \setcounter{equation}{0}  
 
\subsection*{Appendix: Bethe equations}

For completeness, we recall the Bethe equations first derived in the
current context by Nekrasov and Shatashvili \cites{NS1,NS2}.  
Here we derive them formally 
as equations for the critical points of 
the integrand in \eqref{IMB}. See e.g.\ \cite{PSZ}  for a discussion
which does not explicitly 
involve integral representation. 

Let 
\begin{equation}
TX = T \left( T^* \Rep(\bv,\bw)\right) - \sum_i 
(1+\hbar^{-1}) \End(V_i) \label{TX}
\end{equation}
be the tangent bundle of $X$ viewed as an 
element of $K_{\bT\times G} (\fR)$. This is a Laurent 
polynomial in $x_{i,k}$ and the characters of $\bT$. 
The negative terms in it reflect the moment map equations and the 
quotient by $G$. 

Let the transformation $\aroof$ be defined by 
$$
\aroof\left( \sum n_i \chi_i\right) = \prod 
\left(\chi_i^{1/2} - \chi_i^{-1/2}\right)^{\textstyle n_i} \,, \quad
n_i \in \Z \,,
$$
where $\chi_i$ are weights of $\bT \times G$. This is a 
homomorphism from the group algebra of the weight lattice to 
rational functions on a double cover of the maximal torus. 

The following is a restatement of a result of Nekrasov and 
Shatashvili \cites{NS1,NS2}. 

\begin{Proposition}\label{pBethe} 
The critical points in the $q\to 1$ asymptotics of the integral 
\eqref{IMB} satisfy the following Bethe equations 
\begin{equation}
\aroof \left( x_{i,k} \frac{\partial}{\partial x_{i,k}} TX \right) =
z_i  \label{Bethe_eq}
\end{equation}
for all $i\in I$ and $k=1,\dots, \bv_i$. 
\end{Proposition}

The exact form of the right-hand side in \eqref{Bethe_eq} depends
on the shift of variable $z$, which was mentioned 
but not made explicit in the discussion of \eqref{IMB}. 
In \cite{pcmi} it is explained why it is natural to use 
$$
z_\# = z \, (-\hbar^{1/2})^{-\det T^{1/2}}
$$
in place of $z$ in \eqref{IMB}, see also \cite{afo}. It is directly 
related to the shift by the canonical theta-characteristic in 
\cite{MO}. With this shift, the equation \eqref{Bethe_eq} 
take the stated form. 

Note that 
\begin{equation}
\det T^{1/2} = \prod_{\chi \in T^{1/2}X} \chi \label{detT12}
\end{equation}
is a line bundle on $X$ and hence a cocharacter of the K\"ahler torus 
$\bZ$. It therefore makes sense to shift the variables 
$z$ by the value of this 
cocharacter at $-\hbar^{1/2}$. 
Concretely, the coordinates of \eqref{detT12} in the lattice of 
cocharacters are the exponents of $x_{i,k}$ in \eqref{detT12}. 
Note that these 
exponents do not depend on $k$.

\begin{proof}[Proof of Proposition \ref{pBethe} ]

Let $\Phi$ denote the term with $\phi$-functions in 
\eqref{IMB}. We recall from \cite{afo} that 
\begin{equation}
\Phi = \prod_{\chi\in T^{1/2}X} \frac{\phi(q \, \chi)}{\phi (\hbar
  \,\chi)}\label{phi_prod}
\end{equation}
where the product is over the weight $\chi$ in a polarization 
$T^{1/2}X$ of \eqref{TX}. By definition of a polarization, we
have 
\begin{equation}
  \label{eq:1}
TX = \sum_{\chi \in T^{1/2}X} \left(\chi + \frac{1}{\hbar \chi}\right)
\,. 
\end{equation}
Approximating a sum by a Riemann integral gives 
$$
\ln \frac{\phi(q \, \chi)}{\phi (\hbar \,\chi)}  
\sim \frac{1}{\ln q} \int_{1}^{\hbar} \ln (1-s \chi) \, \frac{ds}{s} 
\,, \quad q\to 1 \,. 
$$
Elementary manipulations give 
\begin{multline}
  x \frac{\partial}{\partial x} \int_{1}^{\hbar} \ln (1-s \chi) \,
  \frac{ds}{s} = - \frac{x \frac{\partial}{\partial x} \chi}{\chi} \, \ln
  \frac {1-\chi}{1-\hbar \chi} = \\ = 
 - \ln \aroof \left( x \frac{\partial}{\partial x}  
\left(\chi + \frac{1}{\hbar \chi}\right) \right) +  
 \ln(-\hbar^{1/2})
\, x \frac{\partial}{\partial x} \ln \chi 
\,. 
\label{form1} 
\end{multline}
Note that summed over $\chi$ 
the first term on the second line 
of \eqref{form1} gives $\ln  \aroof \left( x
  \frac{\partial}{\partial x}  TX \right)$. 

The other exponentially large term in \eqref{IMB} is 
$\be(x,z_\#)$, where $z_\#$ denotes the K\"ahler 
variables $z$ shifted by $(-\hbar^{1/2})^{-\det T^{1/2}}$, as above. 
By definition, this means 
\begin{equation}
  \label{form2}
  x_{i,k} \frac{\partial}{\partial x_{i,k}}
\ln \be(x,z_\#) = \frac{1} {\ln q} 
\left(\ln z_i - \ln(-\hbar^{1/2}) \sum_\chi 
x_{i,k} \frac{\partial}{\partial x_{i,k}} \ln \chi 
\right)\,. 
\end{equation}
Summing \eqref{form1} and \eqref{form2} gives
$$
\ln  \aroof \left( x_{i,k} 
  \frac{\partial}{\partial x_{i,k} }  TX \right) = \ln z_i 
$$
as equations for the critical point of the function 
$\cW$ in \eqref{cW}, 
as claimed. 

\end{proof}

\end{document}